\documentclass{article}
\usepackage{graphicx,fancyhdr,amsmath,amssymb,amsthm,subfig,url,hyperref}
\usepackage[margin=1in]{geometry}
\usepackage[shortlabels]{enumitem}
\usepackage{bbm}
\usepackage{wrapfig}
\usepackage{xcolor}
\usepackage[ruled,vlined,linesnumbered]{algorithm2e}
\usepackage[full]{complexity}
\usepackage{todonotes}

\graphicspath{{/figures/}}

\theoremstyle{definition}
\newtheorem{defin}{Definition}
\newtheorem{lemma}{Lemma}
\newtheorem{theorem}{Theorem}

\newtheorem{corollary}[defin]{Corollary}
\newtheorem{claim}{Claim}
\newtheorem{proposition}{Proposition}
\newtheorem{property}{Property}

\newtheorem{innercustomprop}{Proposition}
\newtheorem{innercustomlemma}{Lemma}

\newenvironment{customprop}[1]
  {\renewcommand\theinnercustomprop{#1}\innercustomprop}
  {\endinnercustomprop}
  
\newenvironment{customlemma}[1]
  {\renewcommand\theinnercustomlemma{#1}\innercustomlemma}
  {\endinnercustomlemma}

\newcommand{\ep}{\epsilon}
\newcommand{\val}{\mathsf{Val}}
\newcommand{\appeal}{\mathsf{Appeal}}
\newcommand{\dmax}{\delta_{\text{max}}}
\newcommand{\emax}{\epsilon_{\text{max}}}
\DeclareMathOperator*{\argmax}{arg\,max}

\fancypagestyle{plain}{}
\newcommand{\diff}{\mathsf{diff}}

\graphicspath{{figures/}}

\title{The Smoothed Complexity of Policy Iteration for Markov Decision Processes}
\date{}
\author{
    Miranda Christ\\
    Columbia University\\
    \texttt{mchrist@cs.columbia.edu}
    \and 
    Mihalis Yannakakis\\
    Columbia University\\
    \texttt{mihalis@cs.columbia.edu}
}
\begin{document}
\maketitle

\begin{abstract}
We show subexponential lower bounds (i.e., $2^{\Omega (n^c)}$) on the smoothed complexity of the classical Howard's Policy Iteration algorithm for Markov Decision Processes. The bounds hold for the total reward and the average reward criteria. The constructions are robust in the sense that the subexponential bound holds not only on the average for independent random perturbations of the MDP parameters (transition probabilities and rewards), but for all arbitrary perturbations within an inverse polynomial range. We show also an exponential lower bound on the worst-case complexity for the simple reachability objective.
\end{abstract}

\newpage

\tableofcontents

\newpage

\section{Introduction}
{\em Markov Decision Processes (MDP)} are a fundamental model for dynamic optimization in a stochastic environment with applications in many areas, including operations research, artificial intelligence, game theory, robotics, control theory, and verification. They were originally introduced by Bellman \cite{Bellman} and have been studied extensively since then; see \cite{Howard,Derman,Puterman} for general expositions.  We will define formally MDPs in Section 2, but we give here an informal brief description. MDPs are an extension of Markov chains with an agent, who can affect the evolution of the chain. An MDP consists of a set of states and a set of possible actions that the agent can take at each state, where each action yields a reward to the agent, and results in a probabilistic transition to a new state. Execution of the MDP starts at some state and then moves (probabilistically) in discrete steps from state to state according to the action selected by the agent in each step. The problem is to find an optimal {\em policy} for the agent, i.e. choice of action in each step, that maximizes a desired objective, such as the expected total reward collected during the execution. Although the agent is allowed in each step to use randomization in their choice of action and to base their decision on the complete past history, it is known that there is always a so-called {\em positional} optimal policy that is deterministic and memoryless, i.e. it depends only on the state and selects a unique action for each state.

In some applications (for example in verification and control theory), the objective is not based on rewards, but rather the goal is to maximize the probability that the execution that is generated satisfies a desirable property (expressed for example in a temporal logic); see e.g. \cite{Vardi85,CY95,CY98,BaierAFK18}. It is known that for a broad range of properties this problem reduces to the case of a simple reachability objective, where the goal is to hit a certain target state in a larger MDP that combines the desired property and the original MDP. The reachability objective can be viewed as a special case of the total reward objective (see Section 2), thus the solution methods for reward-based MDPs can be used also for the class of applications that seek to optimize the probability of a desirable execution.

MDPs can be solved in polynomial time using Linear Programming.  From an MDP, one can construct a Linear Program (LP), whose basic feasible solutions (bfs) correspond to positional policies of the MDP, and the optimal bfs yields the optimal policy. The usual way however of solving MDPs in practice is using the {\em Policy Iteration} (PI) algorithm of Howard \cite{Howard}. This is essentially a local search algorithm, an iterative algorithm which starts with an initial positional policy, and keeps improving it until it arrives at an optimal (positional) policy. In each iteration, the algorithm
computes the value for each state according to the current policy and
 determines whether switching the selected action for a state would improve its value; if there are such switchable states, then their actions are switched to obtain the new policy, otherwise the policy is optimal, i.e., in this case local optimality guarantees global optimality. If at some point there are multiple switchable states,
and/or multiple choices of a new action that improves the value for a state, then there is flexibility on which states the algorithm chooses to switch and to which actions, resulting in different versions of Policy Iteration. The most commonly used version, called {\em Howard's PI} (or {\em Greedy PI}), switches simultaneously all switchable states to their most ``appealing" action (see Section 2 for the formal definition). At the other extreme one may choose to switch only one of the switchable states, where the choice of the state and the new action is based on some criterion. We refer to these choices as pivoting rules, in analogy with the Simplex algorithm. Indeed, there is a close correspondence between the variants of Policy Iteration where only one state is switched in each iteration and Simplex applied to the LP for the MDP. Howard's PI corresponds to performing simultaneously many Simplex pivots. 

The (worst-case) time complexity of Howard's PI was open for a long time, until it was finally resolved by Fearnley in \cite{fearnley2010exponential}, who showed an exponential lower bound under the total reward and the average reward objectives.
This was extended to the discounted reward objective in \cite{hollanders2012complexity} for discount factors
that are exponentially close to 1 (in discounted reward MDPs, future rewards are discounted by a given discount factor $\gamma <1$).
For constant discount factor $\gamma$ however, or even if $1-\gamma > 1/poly$, Howard's PI runs in strongly polynomial time \cite{Ye}; this holds more generally even in 2-player turned-based stochastic games for the analogous strategy improvement algorithm \cite{HansenMZ}.
The complexity of PI where only one state is switched in each iteration was studied earlier by Melekopoglou and Condon \cite{melekopoglou1994complexity}, who gave exponential lower bounds for several pivoting rules. More recently, the close connection between single-switch Policy Iteration for MDPs and Simplex for LPs has been exploited to show exponential or subexponential lower bounds for Simplex under various open pivoting rules, by first showing the results for MDPs and then translating them to Simplex: this was shown for the Random-Facet and the Random-Edge rules in \cite{friedmann2011subexponential}, for 
Cunningham's rule in \cite{Cunningham-rule}, and for Zadeh's rule in \cite{Zadeh-rule}. There is ongoing extensive literature on the complexity of Policy Iteration, studying various variants (e.g. randomized PI, geometric PI etc.), special cases (e.g. deterministic MDP) and/or  improving the bounds \cite{scherrer13,hollanders16,taraviya19,wu22}.

Thus, although PI runs fast in practice, its worst-case complexity is exponential for Howard's PI, as well as other variants. This is similar to the behavior of the Simplex algorithm, and more generally a host of other local search algorithms for various optimization problems. To provide a more realistic explanation for the observed performance of Simplex, Spielman and Teng introduced the {\em smoothed analysis} framework \cite{spielman2004smoothed}, a hybrid between worst-case and average-case analysis. 
On one hand, average-case complexity is an algorithm's expected runtime given a probability distribution over inputs. On the other hand, we can think of worst-case complexity as the maximum of an algorithm's expected runtime over all input distributions, including those with all probability mass on a single input.
The smoothed complexity of an algorithm is its maximum expected runtime over all input distributions \textit{with some smoothness constraint}. 
For example, an input is picked arbitrarily by an adversary and then
its parameters (for example the entries of the matrix in LP,
the rewards and transition probabilities in an MDP) are perturbed randomly according to a distribution with density function bounded by a parameter $\phi$ (for example, uniform in $[-\phi,\phi]$, Gaussian or some other distribution). The smoothed complexity of the algorithm then is the expected running time as a function of the input size $n$ and $\phi$. Ideally we would like to have polynomial time in $n$ and $\phi$. Note this is useful if $\phi$ is polynomially bounded in $n$ (or constant), because for exponentially large $\phi$ (i.e. perturbations that are sharply concentrated), polynomial time in $n$ and $\phi$ is simply exponential time, which is not useful.
Smoothed analysis may capture runtime in practice more effectively than worst-case analysis, especially when the numerical values in the input may have some natural variation, as problems formulated from the real world often do.  Spielman and Teng showed that the Simplex algorithm under a certain pivoting rule has polynomial smoothed complexity  \cite{spielman2004smoothed} (and there is a series of subsequent papers simplifying the proof and improving on the bounds, eg. \cite{deshpande05,dadush20}).  

Smoothed analysis has since been applied to a range of problems in areas such as mathematical programming, machine learning, numerical analysis, etc. \cite{spielman2009}. In the area of combinatorial optimization, it has been applied to local search algorithms for problems such as the Traveling Salesperson Problem (TSP), Max-Cut and others. It has been shown for example that the simple 2-Opt algorithm for TSP has polynomial smoothed complexity \cite{TSPsmoothed}, in contrast to its worst-case exponential complexity \cite{Lueker}. For Max-Cut, the simple Flip algorithm has smoothed complexity that is at most quasi-polynomial for general graphs \cite{etscheid2017smoothed,chen2020smoothed} and polynomial for the complete graph \cite{angel2017local,bibak2019improving}, again in contrast to its worst-case exponential complexity \cite{schaffer1991simple}.

Given the good empirical performance of PI and its relationship to the Simplex algorithm, it is natural to hypothesize that the smoothed complexity of PI may well be also polynomial.
Note that this does not follow from the result for Simplex, despite their strong connection, for various reasons.
First, in the smoothed model for Linear Programming all the numerical parameters are randomly perturbed independently. In the MDP, we want to perturb similarly the rewards and transition probabilities, however we want the perturbed model to be also an MDP, in particular the transition probabilities for each action must sum to 1.
Second, in the LP smoothed model, all entries of the constraint matrix are perturbed randomly, even those that are 0; if we apply such perturbation to
the LP of an MDP, it will have the effect of introducing arbitrary
new transitions that have no justification.
In defining the smoothed model for an MDP, it is more natural to preserve the structure of the MDP
(i.e. available actions at each state and possible transitions for each action), since there are usually constraints in the application that is modeled by the MDP that determine which transitions can or cannot occur from a state for each action.
On the other hand, the rewards and transition probabilities may well be estimates, and thus for them it is reasonable to allow perturbations.
Thus, in our smoothed model for MDP, 
we preserve the structure of the MDP, and allow perturbations of the (nonzero) transition probabilities and rewards.

In the literature on smoothed complexity, both models have been used, the {\em full perturbation} model, where all numerical parameters are perturbed, including those that are 0, and what we may call the {\em structured model}, where only the nonzero parameters are perturbed and the structure of the input is preserved. 
For example the analysis of Simplex uses the full perturbation model. Work on local search algorithms for combinatorial optimization have used both models.
For example, in the case of the FLIP algorithm for Max Cut, 
\cite{angel2017local,bibak2019improving} use the full perturbation model and show that the smoothed complexity is polynomial. On the other hand, \cite{etscheid2017smoothed,chen2020smoothed} use the structured model and show that smoothed complexity is quasi-polynomial for every graph; note that the full model coincides with the structured model in the special case when the input graph is complete.
Although the structured part of the input (the graph) is not perturbed, thus it allows for arbitrarily complex, "pathological" instances, the smoothening of the numerical parameters (the edge weights) brings the complexity down from exponential to quasi-polynomial; the conjecture is in fact that the true smoothed complexity is polynomial.

Depending on the application, one or the other model may be more reasonable. In the case of MDPs, we believe that the structured model is more natural for the reasons discussed above. The MDP typically models an application at hand (for example, a probabilistic program that is analyzed, a control design problem, a game etc.), and the transitions have some meaning in the application. The precise values of probabilities and rewards may be fungible, but their existence is important. Changing the structure of the instance changes the problem, or may even render it meaningless.

For example, consider an MDP with a reachability objective.
If perturbations are applied also to the zero-probability (i.e. nonexistent) transitions, then in the perturbed MDP every possible transition between any two states will be included with nonzero probability. This means that for any policy the graph of the MDP becomes strongly connected, every policy will reach the target with probability 1, and the problem has disappeared.

The issue of preserving the (zero-nonzero) structure is especially important when the model is used to formulate and solve other problems. For a simple example, consider the following: MDPs with
rewards can be used to solve the simple reachability optimization problem for MDPs (without rewards), since the latter can be viewed as a special case (can be reduced to) the former:
All transitions of the reachabilty MDP are given 0 reward, except for the transitions into the target state that are given reward 1; maximizing the expected total reward in the resulting MDP is equivalent to maximizing the probability of reaching the target state in the reachability MDP.
If in the reward MDP we are allowed to perturb the zero rewards then the problem has changed, and optimization in the MDP with rewards no longer correctly captures the MDP reachability problem.

\subsection{Our Results}

In this paper we study the smoothed complexity of Policy Iteration.
Given its similarity to the Simplex algorithm, one might hope to show polynomial smoothed upper bounds for PI.
We show the contrary: for several prominent policy iteration variants, such a result is impossible; the smoothed complexity is subexponential or even exponential. We concentrate here mainly on the total reward objective. 

Our main result concerns the classical Howard's (Greedy) PI which switches simultaneously all switchable states to their actions with greatest appeal.
We show that Howard's PI has at least subexponential smoothed complexity under the total reward objective; a similar result holds for the average reward objective. Furthermore, the lower bound holds not only for
the expected complexity under random independent perturbations of the parameters,  but it holds in fact
for {\em all (arbitrary)} perturbations within a certain inverse polynomial range. 
(The amount of perturbation corresponds to the $1/\phi$ parameter of the smoothed model, so to be meaningful, $\phi$ has to be at most polynomial.)
Specifically, we construct an MDP with $N$ states and bounded parameters (rewards and transition probabilities), such that in every MDP obtained by perturbing the parameters by any amount up to $1/N$, Howard's PI requires at least $2^{\Omega(N^{1/3})}$ iterations.

Our initial approach for this was to examine whether the 
construction of \cite{fearnley2010exponential} for the worst-case complexity can be modified to prove a smoothed lower bound. However, we were not able to do this. Unfortunately, the construction seems to be brittle and does not hold up under perturbations.
Thus, we started fresh and designed a new construction with robustness in mind.
The construction and the proof are quite involved. This is to be expected, considering that the construction of \cite{fearnley2010exponential} was also quite intricate. That construction involved positive and negative rewards, exponentially small probabilities, and exponentially large rewards. We show that the parameters do not need to be exponentially large or small, and furthermore they can tolerate arbitrary perturbations up to an inverse polynomial, without affecting the behavior of Howard's PI algorithm. 

Furthermore, we use the robustness of our construction for MDP with rewards, to show that the worst-case complexity of Howard's PI for MDPs with the simple reachability objective is exponential. Note that these MDPs have no rewards 
(or as mentioned above they are a special case of MDPs with rewards 0 and 1).
In some sense, this second construction is an approximate reduction from MDPs with rewards to the special case of reachability MDP.
The robustness of the original reward MDP is essential to establish the correctness of the result for the weaker reachability MDP.

We also analyze three simple variants of PI from \cite{melekopoglou1994complexity} that switch a single (switchable) state in each iteration, chosen according to some rule.
In {\em Simple PI} the state is chosen according to an arbitrary initial priority order; in {\em Topological PI} it is chosen according to a topological order; and in {\em Difference PI} it is chosen
according to the difference in value between the new and the old
action of the state; see Section 2 for a formal definition of the variants.
We make slight modifications to the constructions from \cite{melekopoglou1994complexity} and prove that they are robust to perturbations. These constructions are  reachability MDPs; thus the only numerical parameters are the transition probabilities, there are no rewards (or equivalently, all the rewards are 0 except for the transitions to the target state that have reward 1).
Simple PI and Topological PI take exponential time, for very large (constant) perturbations of the transition probabilities. Difference PI takes at least subexponential time for inverse polynomial perturbations. 

We finally discuss the relationship between our results for the Single switch PI variants and the Simplex algorithm, describing precisely how our perturbations of an MDP translate to the corresponding LP.
We state the lower bounds implied by our Policy Iteration results for Bland's and Dantzig's pivot rules in the Simplex algorithm for LPs arising from MDPs, though these bounds are not new for Simplex.

\subsection{Outline of proof techniques}
The construction and proof of the main result on Howard's (Greedy) PI are quite complex and involved. We first design a new construction of an MDP $M_1$ for the exponential worst-case complexity of Greedy PI, which is more amenable to modifications to achieve the desired robustness. As is usual in exponential lower bounds for many problems, the MDP is constructed so that the iterations of Greedy PI will simulate a binary counter counting from  0 up to $2^n$. There is a set of $n$ states $b_i$ of the MDP (among many others), each with two distinguished actions 0, 1, where the choices of states $b_i$ correspond to the bits of the counter. The MDP $M_1$ is constructed so that if in the initial policy all states $b_i$ choose action 0, then Greedy PI will go through $2^n$ rounds until it arrives at the optimal policy where all states $b_i$ choose action 1. Each round involves a number of steps. To manage the complexity of the construction, we build it in stages.
We first design a simpler MDP $M_0$, which exhibits this exponential (worst-case) behavior for a slight variant of Greedy PI, call it Hybrid PI, which has an additional rule that a switchable state $b_i$ can switch from action 0 to 1 only if all $b_j$ with $j<i$ have chosen action 1,
and no other non-$b_i$ states are switchable.
We then modify $M_0$ to an MDP $M_1$ by using a suitable gadget at the states $b_i$ which serves the purpose of delaying the switches at states $b_i$ when running Greedy PI on $M_1$ in such a way that it behaves like Hybrid PI on $M_0$.
As a result, Greedy PI on $M_1$ simulates a binary counter and takes exponential time.

The MDP $M_1$ has exponentially large and small rewards (both positive and negative), and exponentially small probabilities. The next stage in the proof transforms $M_1$ to another MDP $M_2$ that has bounded rewards and probabilities, and which is robust in the sense that Greedy PI has the same behavior for any perturbation of the rewards and probabilities up to an inverse polynomial amount.
This transformation is done using appropriate gadgets. We design gadgets to simulate exponentially large and exponentially small rewards and transition probabilities, and ensure that the gadgets are robust, i.e., they perform correctly (approximately) even under perturbation of their rewards and probabilities.  Finally, we ensure that the analysis for $M_1$ is robust enough, so that the behavior of Greedy PI on it is simulated by $M_2$ even under perturbation of its parameters.

The proof for the worst-case exponential complexity of Greedy PI under the reachability objective uses the constructed MDP $M_1$ for the total reward objective (with somewhat modified parameters). The MDP $M_1$ has positive and negative rewards, whereas there are no rewards in the reachability objective.
We design suitable gadgets to eliminate positive and negative rewards using random actions, and apply them to transform $M_1$ to a new MDP $M_3$, without rewards, for the reachability problem. An important requirement for the correct functioning of the gadgets is that we must know bounds on the minimum and maximum value of the nodes where the gadgets are plugged in, which we have from the analysis of $M_1$. The robustness of the MDP $M_1$ is critical for the correctness of the transformation,
i.e. so that the behavior of Greedy PI under the reachability objective in $M_3$ simulates the behavior of Greedy PI in $M_1$ under the total reward objective.

The proofs for the results on the variants of PI with a single switch use the constructions of \cite{melekopoglou1994complexity}, sometimes with some small modifications. The proofs are relatively simple and offer a gentle introduction to the issues, and the unfamiliar reader might like to read this section first. In the case of Simple PI and Topological PI, the analysis follows closely that of \cite{melekopoglou1994complexity}, except that it is carried out for general values of the transition probabilities, rather than specific values. In the case of Difference PI, we use a parameterized gadget with suitable choice of parameters to modify the construction in such a way that it can tolerate perturbations of the transition probabilities within an inverse polynomial range.

\medskip

\noindent{\bf Organization of the paper.}
The rest of the paper is organized as follows.
Section 2 gives basic definitions and notation.
Section 3, which is the heart of the paper, shows that Greedy (Howard's) PI has at least subexponential smoothed complexity under the total reward objective. 
Section 4 builds on our construction to show the worst-case exponential complexity of Greedy PI under the simple reachability objective.
Section 5 presents the results for three PI variants with single state switch: Simple PI, Topological PI and Difference PI, and Section 6 notes the connection to Simplex pivoting rules.

\section{Preliminaries}
A Markov Decision Process consists of a (finite) set of states $S$, and a (finite) set $A_s$ of available actions for each state $s \in S$. Let $A = \cup_{s \in S} A_s$ denote the set of all actions.
For each action $a \in A_s$ there is a probability distribution of the state(s) resulting when taking action $a$ at state $s$ that is described by a function $p : S \times S \times A \to \mathbb{R}^+$, where $p(s' | s, a)$ denotes the probability of ending up at state $s'$ when taking action $a$ from $s$.
The action $a$ is {\em deterministic} if 
$p(s' | s, a)=1$ for some $s'$  and 
$p(s" | s, a)=0$ for all other $s"$.
Each action yields some (possibly zero) reward, represented by a function $r : S \times A \to \mathbb{R}$ where $r(s, a)$ denotes the reward obtained by taking action $a \in A_s$ from $s$.
 A (positional) \emph{policy} is a function $\pi : S \to A$, where for each state $s \in S$, $\pi(s) \in A_s$ is the action selected at that state. A policy $\pi$ for an MDP $M$ induces a Markov chain $M_{\pi}$ on the same state set $S$, where the transition probabilities out of each state $s$ are given
 by $p(s'|s, \pi(s))$.
 
 A \emph{criterion} (or {\em objective}) is a function that, given a policy, associates a value with each state. 
We consider primarily the \emph{total reward} criterion, which yields the following notions of value and appeal. The value of a state $s$ captures the expectation of the sum of rewards accrued by starting at $s$ and taking the actions given in the policy as time goes to infinity; it is well-defined for MDPs where one must eventually reach a sink state, one that has no actions and no outgoing transitions.
 Under the total reward criterion, the \emph{value} of a state $s$ under a policy $\pi$ satisfies the equation
$$\val^\pi(s) = r(s, \pi(s)) + \sum_{s' \in S} p(s' | s, \pi(s)) \cdot \val^\pi(s')$$

Given policy $\pi$, the \emph{appeal} of an action $a$ at $s$ is
$$\appeal^\pi(s, a) = r(s, a) + \sum_{s' \in S} p(s'|s,a) \cdot \val^\pi(s')$$

 An {\em optimal policy} is a policy that maximizes the value of every state (there is always such a policy).
 
We later consider the \emph{reachability} criterion, where the goal is to maximize the probability of reaching a given target sink state $s^*$. In this case, the
value of a state $s$ under a policy $\pi$ is the probability of reaching the target $s^*$  following the actions selected in $\pi$.
The reachability criterion can be viewed as a special case of the total reward criterion, by assigning reward zero to all transitions except for those going from other states into the target state $s^*$, which are assigned reward 1. An {\em optimal policy} is a policy that maximizes the value of every state.  

{\em Policy Iteration} or {\em Policy Improvement} (PI) is a family of local search algorithms used to find an optimal policy of an MDP. We say a state $s$ is \emph{switchable} under a current policy $\pi$ if there is an action $a \in A_s$ such that $\appeal^\pi(s, a) > \val^\pi(s)$. 
We also say any such value-improving action $a$ is switchable.
In each iteration, PI switches some number of switchable states to their value-improving actions. An optimal policy is reached when no states are switchable.
There are several PI variants, which involve various switching rules for choosing the state(s) and actions(s) to switch in each iteration.

The most widely used variant, {\em Howard's PI}  (or \emph{Greedy PI}) involves switching \emph{all} switchable states in each iteration.
A switchable state with multiple switchable actions is switched to the action with greatest appeal.
More precisely, given that the current policy is $\pi$, Greedy PI switches each switchable state $s$ to an action in $\argmax_{a \in A_s} \appeal^\pi(s, a)$.

We discuss also several variants of PI that switch only one state in each iteration,  \emph{Simple PI}, \emph{Topological PI}, and \emph{Difference PI}.
Simple policy iteration fixes an ordering over the states and switches the highest-numbered switchable state. If there are multiple improving actions at a state, one of them is chosen according to some rule;
in the constructions we discuss, every state has only two actions, so there is no choice of improving action.
Topological policy iteration considers a topological ordering over the states, where if there is a path (sequence of actions with nonzero probability) from a state $i$ to a state $j$, the order of $i$ is at least the order of $j$;
that is, the graph is partitioned into strongly connected components and each state is assigned the index of its component in a topological order. Topological PI switches the highest-numbered switchable state of the component with lowest topological order that contains switchable states.
Difference policy iteration switches the switchable state with the greatest difference between the appeal of the action it switches to its current value.

We represent MDPs graphically, where states are vertices and actions are directed edges. We sometimes use this terminology in our discussion.
Each deterministic action is shown as a directed edge between two nodes.
Each probabilistic action is shown as starting as a single line at the origin state and branching into multiple lines to reach the various possible resulting states. 
Given a policy $\pi$, the set of directed edges corresponding to the selected actions form a subgraph of the MDP (this is the graph of the Markov chain $M_{\pi}$). 
That is, a node $s$ selecting action $a$ under $\pi$ has a directed edge to every node $s'$ that $a$ takes $s$ to with nonzero probability.
We say a node $s'$ is \emph{reachable} from $s$ if there exists a path from $s$ to $s'$ in this subgraph.

\paragraph{Smoothed model.} 
The smoothed analysis framework lies between average-case analysis and worst-case analysis. It considers input instances with each parameter drawn independently from some probability distribution (e.g., Gaussian, uniform or any other distribution) with an upper bound $\phi$ on its density function.
An algorithm $A$ has polynomial smoothed complexity if the maximum expected runtime of $A$ over all such distributions is polynomial in both $\phi$ and in the size of the input.
Alternatively, before $A$ is given an arbitrary (worst-case) input $x$,
the values of $x$ (the numerical parameters) are perturbed according to some distribution still of bounded density at most $\phi$.
The perturbed input $x'$ is then given to $A$.
The smoothed runtime of $A$ is its worst-case (over all inputs $x$) expected runtime (expected over the perturbation distribution).

As discussed in the Introduction, we consider the structured perturbation model for MDPs, where we perturb only the nonzero transition probabilities and rewards.
Our constructions are robust not only to random perturbations, but moreover to {\em all}  perturbations within a certain wide range.
When proving our lower bounds, we model the nonzero rewards and the probabilities associated with probabilistic actions as adversarially chosen within some perturbation radius $\sigma = 1/\phi$, where the adversary aims to minimize the expected runtime (i.e. to defeat the lower bound construction). 
That is, any reward $r \neq 0$ of $x$ can take on any value $r'$ in $[r - \sigma, r + \sigma]$ in $x'$. We do not perturb rewards of zero.
For any probabilistic action with nonzero transition probabilities $p_1, \ldots, p_k$, the corresponding perturbed probabilities $p'_i$ in $x'$ are any (non-negative) values in $[p_i - \sigma, p_i + \sigma]$ that sum to 1.
Note that we cannot perturb the probabilities independently, since they must sum to 1 for the resulting MDP to be valid.
(An alternative, equivalent model, is to perturb independently all the $p_i$ within the allowed range and normalize them so they sum to 1.)
We say such an $x'$ is within \emph{perturbation radius} $\sigma$ of $x$.

Each lower bound in this paper gives an MDP $x$ where \emph{every} MDP within perturbation radius $\sigma$ of $x$ yields superpolynomial runtime. 
Thus, our results are stronger than the usual smoothed analysis: the superpolynomial lower bounds hold not only for random perturbations according to a specific probability distribution of (inverse polynomial) bounded density, but they moreover hold for {\em all} perturbations within the specified ranges, i.e., even when an adversary who wants to defeat the construction and minimize the running time chooses any perturbations they want within the specified range.

\section{A smoothed lower bound for Greedy PI under the total reward and average reward criteria}

In this section, we prove a subexponential lower bound on the smoothed complexity of Greedy PI. More specifically, we construct an MDP with $N$ states and bounded parameters (rewards and transition probabilities), such that in every MDP obtained by perturbing the parameters by any amount up to $1/N$, Greedy PI requires at least $2^{\Omega(N^{1/3})}$ iterations. We prove this for the total reward criterion. The same result applies to the average reward criterion.

The construction is quite involved and is presented in several stages. We present first in Section 3.1 a simplified construction which forces exponential worst-case runtime for a variation of Greedy PI (we call it hybrid Greedy PI), in which in certain cases some switchable states are not switched until some conditions are satisfied.
In Section 3.2, we add suitable gadgets to this MDP so that Greedy PI in the new MDP simulates the hybrid variant in the simplified construction; thus, Greedy PI has exponential worst-case runtime in this full construction.
This MDP includes rewards that are exponentially large and small, and some probabilities that are exponentially small.
In Section 3.3 we transform this MDP to our final robust MDP by using gadgets that allow us to eliminate the exponentially large and small rewards and probabilities and simulate them by parameters that lie in a bounded range in a robust way; that is, the behavior of Greedy PI is not affected by perturbation of the parameters up to an inverse polynomial amount. 

\subsection{Simple construction}
In this section, we present the simplified construction shown in \autoref{fig:base_mdp_greedy}, on which a variant of greedy policy iteration takes exponentially many iterations for the total reward criterion. 
This variant, which we call \textit{hybrid policy iteration} and which we define for this construction only, is nearly greedy policy iteration, except at the nodes $b_i$. 
We define hybrid PI and the simple construction for ease of presentation, and later in \autoref{subsec:full} we present our full construction, where we use gadgets to ensure that greedy PI behaves similarly to hybrid PI.

\begin{defin}[hybrid policy iteration]
Given a policy $\pi$, hybrid policy iteration chooses the next policy $\pi'$, where:
\begin{itemize}
    \item Every switchable non-$b_i$ vertex is switched to its appeal-maximizing action, as in greedy PI.
    \item Every switchable $b_i$ vertex with $\pi(b_i) = 1$ is switched so that  $\pi'(b_i) = 0$.
    \item Every switchable vertex $b_i$ with $\pi(b_i) = 0$ switches if and only if no non-$b_i$ vertices are switchable and $\pi(b_j) = 1$ for all $j < i$.
\end{itemize}
\end{defin}

The simple construction is shown in \autoref{fig:base_mdp_greedy}, with parameters as follows. Let $r : [n] \to \mathbb{Z}$ be a function where $r(i)$ is the reward on the edge associated with taking action 1 from $b_i$. $r$ need only satisfy that for all $i$, $r(i) > \sum_{j < i} r(j)$. We can achieve this by letting $r(i) = 2^{2i}$. Let $-\ep$ denote a very small negative reward. Assume that $\ep \ll r(i)$ for all $i$. 

We describe the available actions at each state, for each $i \leq n$:
\begin{itemize}
    \item $b_i$: $b_i$ has a deterministic action 0 to $w_{i+1}$ with reward 0 and a deterministic action 1 to $d_{i+1}$ with reward $r(i)$.
    \item $c_i$: $c_i$ has a deterministic action 0 to $w_{i+1}$ with reward 0 and a deterministic action 1 to $b_{i}$ with reward $-\epsilon$.
    \item $d_i$: $d_i$ has a deterministic action 0 to $w_{i+1}$ with reward 0 and a deterministic action 1 to $c_i$ with reward $-\epsilon$.
    \item $w_i$: $w_i$ has a deterministic action 1 to $b_i$ with reward 0. It also has a deterministic action to $b_j$ for every $j > i$ with reward 0, and a deterministic action to the sink with reward 0.
\end{itemize}
$w_{n+1}$ has only a deterministic action to the sink node with reward 0. $d_{n+1}$ has a deterministic action to the sink with reward $-\ep$. $c_{n+1}$ has a deterministic action to $d_{n+1}$ with reward $-\ep$. 

At a high level, our construction simulates a binary counter. Each state $b_i$ represents a bit. When $b_i = 1$, a large reward is incurred when leaving $b_i$. When $b_i = 0$, no reward is incurred. The starting policy for our lower bound will have all bits $b_i$ set to 0, and the optimal policy is when all bits $b_i$ are set to 1. The following two properties ensure that the bits behave as a binary counter, iterating through all binary strings of length $n$ before reaching the optimal policy. We state them here and prove them later. We achieve Property 1 by construction, and we achieve Property 2 by definition of hybrid PI.

\begin{property}
When a bit $b_i$ is set to 1, all lower bits $b_j$ for $j < i$ are reset to 0 within two iterations.
\end{property}

\begin{property}
For each $i$, $b_i$ switches to 1 only after $b_j = 1$ for all $j < i$.
\end{property}

\begin{figure}
    \centering
    \includegraphics[width=200pt]{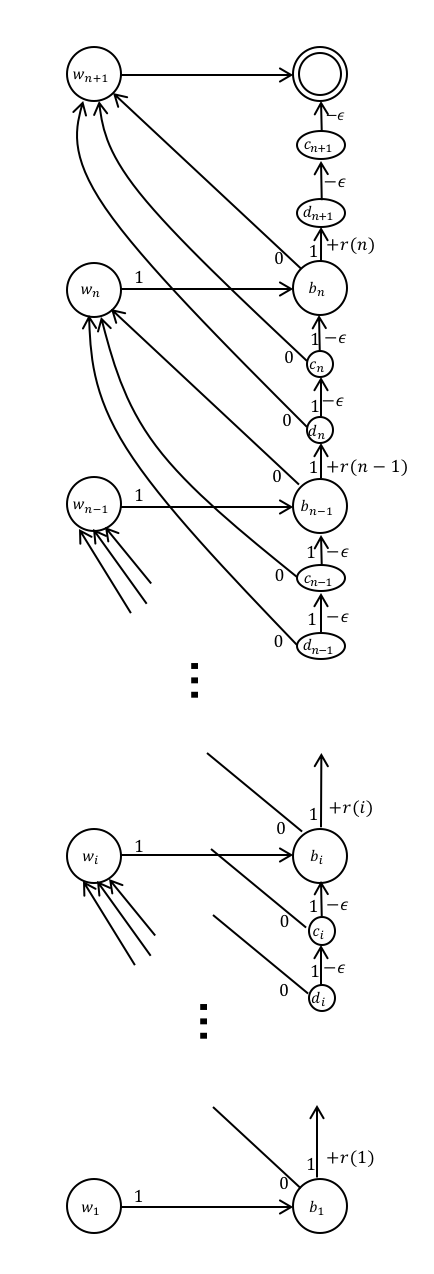}
    \caption{The base MDP for our greedy PI construction. The sink is shown with a double border. Each $w_i$ has a deterministic edge with reward 0 to every $b_j$ for $j \in [i, n]$, and a deterministic edge to the sink. These edges are omitted in the figure for clarity.}
    \label{fig:base_mdp_greedy}
\end{figure}

We will show using these properties that hybrid policy iteration proceeds in three phases. We will show that between every set of three phases (i.e., before the first phase and after the third phase), the following invariant always holds. Let $B = \{i | b_i = 1\}$. 

\paragraph{Invariant.} For all $i \in B$, $w_i = b_i = c_i = d_i = 1$ for all $i \in B$.
 For all $i \notin B$, we have $b_i = c_i = d_i = 0$. For all $i \notin B$ and $i > \max (B \cup \{0\})$, $w_i$ chooses the deterministic action to the sink.
 Otherwise, if $i \notin B$ and $i \leq \max B$, $w_i$ chooses the deterministic action to $b_\ell$ where $\ell$ is the smallest index such that $\ell \in B$ and $\ell \geq i$.

\paragraph{Phases.} We now describe the phases. Each set of 3 phases involves adding the minimum index $i = \min( [n] \setminus B)$ to $B$ and resetting all lower indices, so that $i$ is the minimum index in $B$.
\begin{enumerate}
    \item $b_i$ switches from 0 to 1.
    \item $w_j$ switches to $b_i$ for all $j \leq i$. $c_i$ switches to 1. 
    \item $b_j$ switches to 0 for all $j < i$. $d_i$ switches to 1. $c_j, d_j$ switch to 0 for all $j < i$. 
\end{enumerate}
At the end of the 3 phases, the invariant is again satisfied.

\paragraph{All-zero policy.} Let $\pi_0$ denote the policy with each $w_i$ choosing the action to the sink and all other nodes' actions equal to 0.
We show that our construction indeed follows this structure when we start with all actions equal to 0, and each $w_i$ choosing the action taking it to the sink. We first prove several useful facts.

\begin{proposition} \label{prop:b_i-rew1}
 When the invariant is satisfied, for every $i \notin B$ and $i < \max B$ we have $\val(b_i) = \val(c_i) = \val(d_i) = \val(b_\ell)$ where $\ell$ is the smallest index such that $\ell \in B$ and $\ell \geq i$.
\end{proposition}

\begin{proof}
Let $i \notin B$. Thus, $b_i = c_i = d_i = 0$ and $\val(b_i) = \val(c_i) = \val(d_i) = \val(w_{i+1})$. First, suppose that $i+1 \in B$. Then $w_{i+1} = 1$, and $\val(b_i) = \val(w_{i+1}) = \val(b_{i+1})$. $i+1$ is the smallest index in $B$ that is at least $i$, so we are done.

Next, suppose that $i + 1 \notin B$. Then $w_{i+1}$ chooses the action to $b_j$, where $j$ smallest index in $B$ that is at least $i+1$. Since $i + 1 \notin B$, $j = \ell$ is also the smallest index in $B$ that is at least $i$. Thus, $\val(b_i) = \val(c_i) = \val(d_i) = \val(w_{i+1}) = \val(b_\ell)$.
\end{proof}

\begin{proposition} \label{prop:b_i-rew2}
When the invariant is satisfied, for every $i \in [n]$ we have $(\sum_{\substack{j \geq i \\ j \in B}} r(j)) - 2(n-i +1)\ep \leq \val(b_i) \leq \sum_{\substack{j \geq i \\ j \in B}} r(j)$.
\end{proposition}

\begin{proof}
First, we observe that for any $b_i$, there is no path to any $b_j$ for $j < i$. Thus, the value of $b_i$ is at most the sum of its reward and the collected rewards from higher bits. 
We collect a reward from a bit only if it is in $B$. Thus, $b_i \leq \sum_{\substack{j \geq i \\ j \in B}} r(j)$.

We prove the lower bound by backwards induction on $i$.
For every $i > \max B$, $b_i$ goes to $w_{i+1}$ with reward 0, which goes to the sink with reward 0, so $\val(b_i) = 0$. Thus, we can use the base case $i = \max B$.
If $i = n$, $\val(b_n) = r(n) - 2\ep$.
Otherwise, $b_i$ takes action 1 to $d_{i+1}$, collecting reward $r(i)$. $d_{i+1}$ takes action 0 to $w_{i+1}$, which takes an action to some higher bit with value 0. Thus, $\val(b_i) = \val(d_{i+1}) + r(i) = r(i)$.
Assume that for fixed $k \geq 0$, the claim holds for all indices $c \geq n - k$.

Let $j = n - k-1$; first suppose that $j \notin B$. By the invariant, $b_j$ selects its 0 action to $w_{j+1}$. If $j+1 \in B$, $w_{j+1}$ takes its action to $b_{j+1}$ with reward 0. Thus, $\val(b_j) = \val(w_{j+1}) = \val(b_{j+1}) \geq (\sum_{\substack{j \geq i \\ j \in B}} r(j)) - 2(n-i +1)\ep$ by the inductive hypothesis. 
If $j+1 \notin B$, $w_{j+1}$ takes its deterministic action to $b_\ell$ where $\ell$ is the smallest index such that $\ell \in B$. Again by assumption, $\val(b_j) = \val(w_{j+1}) = \val(b_{j+1}) \geq (\sum_{\substack{j \geq i \\ j \in B}} r(j)) - 2(n-i +1)\ep$.

Now, let $j \in B$, so $\val(b_j) = \val(d_{j+1}) + r(j)$. First, suppose that $j+1 \notin B$. By \autoref{prop:b_i-rew1}, $\val(d_{j+1}) = \val(b_\ell)$, where $\ell$ is the smallest index in $B$ that is at least $j$. By the inductive hypothesis, $\val(b_\ell) \geq (\sum_{\substack{a \geq j \\ a \in B}} r(a)) - 2(n-\ell+1)\ep$.
Thus, $\val(b_j) = \val(b_\ell) \geq (\sum_{\substack{a \geq j \\ a \in B}} r(a)) - 2(n-j+1)\ep$ as desired.

Suppose now that $j + 1 \in B$. Thus, $b_{j+1} = c_{j+1} = d_{j+1} = 1$, so $\val(b_j) = \val(b_{j+1}) + r(j) - 2\ep$. By the inductive hypothesis, $\val(b_{j+1}) \geq (\sum_{\substack{a \geq j+1 \\ a \in B}} r(a)) - 2(n-j)\ep$.
Thus, $\val(b_j) \geq r(j) - 2\ep + (\sum_{\substack{a \geq j+1 \\ a \in B}} r(a)) - 2(n-j)\ep \geq (\sum_{\substack{a \geq j \\ a \in B}} r(a)) - 2(n-j+1)\ep$ as desired.

We have thus shown that assuming that the claim holds for all indices $c \geq n-k$, the claim holds for index $j = n - k - 1$. The proposition follows by induction.

\end{proof}

\begin{proposition} \label{prop:b_i-rew3}
Let $b_i = 1$. Then when the invariant is satisfied, $\val(b_i) \geq \val(b_\ell) + r(i) - 2\ep n$ for all $\ell > i$.
\end{proposition}

\begin{proof}
By \autoref{prop:b_i-rew2}, $\val(b_i) \geq (\sum_{\substack{j \geq i \\ j \in B}} r(j)) - 2(n-i +1)\ep \geq (\sum_{\substack{j \geq i \\ j \in B}} r(j)) - 2 \ep n$, and $\val(b_\ell) \leq \sum_{\substack{j \geq i+1 \\ j \in B}} r(j) = (\sum_{\substack{j \geq i \\ j \in B}} r(j)) - r(i)$.
Thus, $\val(b_i) - \val(b_\ell) \geq r(i) - 2\ep n$.

\end{proof}

\begin{lemma} \label{lemma:phases}
When the invariant is satisfied, the set of switchable nodes is exactly the set of bits not in $B$.
\end{lemma}

\begin{proof}
We first show that all bits $b_i$ where $i \in B$ are not switchable.
Let $i \in B$. By \autoref{prop:b_i-rew2}, $\val(b_i) \geq (\sum_{\substack{j \geq i \\ j \in B}} r(j)) - 2(n-i +1)\ep$. 
Switching bit $b_i$ to action 0 would change its value to $\val(w_{i+1})$.
If $i = \max B$, $\val(w_{i+1}) = 0$.
Otherwise, $\val(w_{i+1}) = \val(b_\ell)$, where $\ell$ is the smallest index such that $\ell \in B$ and $\ell \geq i+1$. Since $\ell > i$, $\val(w_{i+1}) \leq \val(b_i) - r(i) + 2\ep n$ by \autoref{prop:b_i-rew3}. Since $r(i) >> \ep$, $b_i$ is not switchable.

We next show that all bits $b_i$ where $i \notin B$ are switchable.
If $n \notin B$, $b_n$ is switchable, since its current value is 0, and it can get reward $r(n) - 2\ep$ by switching to 1.
Fix $i \notin B, i \neq n$. First, suppose that $b_{i+1} = 1$. Then, by the invariant, $w_{i+1} = 1$ and $\val(b_i) = \val(w_{i+1}) = \val(b_{i+1})$. Again by the invariant, $d_{i+1} = c_{i+1} = 1$, so $\val(d_{i+1}) = \val(b_{i+1}) - 2\ep$. Thus, $\appeal(b_i, 1) = \val(b_{i+1}) + r(i) - 2\ep$. Since $r(i) >> \ep$, $\appeal(b_i, 1) > \val(b_{i+1}) = \val(b_i)$ and $b_i$ is switchable to 1.

Suppose in the other case that $b_{i+1} = 0$, so $\val(b_{i}) = \val(w_{i+1})$. If $i+1 > \max B$, both $w_{i+1}$ and $w_{i+2}$ have selected the action to the sink. Otherwise, $w_{i+1}$ and $w_{i+2}$ both have selected the action to the smallest index $\ell$ such that $\ell \in B$ and $\ell \geq i+2$.
Thus, $w_{i+1}$ and $w_{i+2}$ select actions to the same vertex, and $\val(w_{i+1}) = \val(w_{i+2})$. Since $d_{i+1} = 0$ by the invariant, $\val(d_{i+1}) = \val(w_{i+2})$. Thus, $\appeal(b_i, 1) = r(i) + \val(w_{i+2}) > \val(w_{i+1}) = \val(b_i)$. Thus, $b_i$ is switchable to 1.

Next, we argue that none of the bits $c_j$ or $d_j$ are switchable. If in case 1, $j \in B$, then $\val(b_j) > \val(w_{j+1})$ with a nontrivial gap, as argued when showing that bits in $B$ are not switchable. Thus $\val(c_j) = \val(b_j) - \ep > \val(w_{j+1})$. Similarly, $\val(d_j) = \val(b_j) - 2\ep > \val(w_{j+1})$.
$\val(w_j) = \val(b_j) > \val(b_i)$ for all $i > j$.

If in case 2, $j \notin B$, then $\val(b_j) = \val(c_j) = \val(d_j) = \val(w_{j+1})$. Switching $c_j$ or $d_j$ would lose them a reward of $\ep$.

Finally, none of the $w_j$ are switchable. 
First, suppose that $j > \max B$. Then by the invariant, all bits greater than $j$ point to the sink and have value 0. Thus, $w_j$ is not switchable. 

Next, suppose that $j \leq \max B$. Then $w_j$ points to some bit $b_\ell$, where $\ell$ is the smallest index such that $\ell \in B$ and $\ell \geq j$. 
By \autoref{prop:b_i-rew3}, $\val(b_\ell) > \val(b_{\ell'})$ for all $\ell' > \ell$. 
By \autoref{prop:b_i-rew1}, for all $\ell'$ such that $j < \ell' < \ell$, $\val(b_{\ell'}) = \val(b_\ell)$. Thus, for all
$\ell' > j$, $\ell' \neq \ell$, $\val(b_\ell') \leq \val(b_\ell)$, so $w_j$ is not switchable.

\end{proof}

\begin{proposition} \label{prop:b_i-rew4}
When the invariant is satisfied, $\val(b_{i}) - 2\ep \leq \val(d_i) \leq \val(b_{i})$ for every $i \leq n$.
\end{proposition}

\begin{proof}
First, suppose that $i = n$. If $n \in B$, $\val(d_n) = \val(b_n) - 2\ep$. If $n \notin B$, $\val(d_n) = 0$, and $\val(b_n) = 0$. 

Suppose that $i \in B, i \neq n$. Then $d_{i}$ and $c_{i}$ both take action 1, so $\val(d_{i}) = \val(b_{i}) - 2\ep$. 

Suppose that $i \notin B, i \neq n$. Then $d_{i}$ and $b_i$ both take action 0 to $w_{i+1}$ with reward 0. Thus, $\val(d_i) = \val(w_{i+1}) = \val(b_i)$.
\end{proof}

\begin{theorem} \label{thm:hybrid}
Given the simple construction and starting policy $\pi_0$, hybrid policy iteration requires at least $2^n$ iterations to reach the optimal policy under the total reward criterion.
\end{theorem}

\begin{proof}
We show that if the invariant is satisfied, the algorithm proceeds in the three phases. For each phase, we argue that the switches made follow hybrid PI.

\paragraph{Phase 1.} By \autoref{lemma:phases}, the set of switchable nodes is exactly the bits not in $B$. By definition of hybrid PI, the only node that switches is $b_i$ where $i$ is the minimum index such that $b_i = 0$. $b_i$ thus switches to 1, and its value is now $\val(b_i) = \val(d_{i+1}) + r(i)$. 

\paragraph{Phase 2.} Switching $b_i$ in Phase 1 affected the value of only $b_i$, since no node selects an action to $b_i$ according to the invariant.  The only nodes that may become switchable in Phase 2 are those with actions to $b_i$; these nodes are exactly $c_i$ and $w_j$ for $j \leq i$. 
Now, $\val(b_i) \geq \val(b_{i+1}) + r(i) - 2\ep = r(i) - 2\ep + \sum_{\substack{j' \geq i+1 \\ j' \in B}} r(j')$. Thus, for any $j > i$, $\val(b_i) > \val(b_j)$. For any $j < i$, by \autoref{prop:b_i-rew2} and the fact that switching $b_i$ affected the value of only $b_i$, we have that 
$$\val(b_j) \leq \sum_{\substack{j' \geq j \\ j' \in B}} r(j') = \sum_{\substack{j' \geq i+1 \\ j' \in B}} r(j') + \sum_{\substack{j \leq j' < i \\ j' \in B}} r(j') < \left(\sum_{\substack{j' \geq i+1 \\ j' \in B}} r(j')\right) + r(i) - 2\ep$$ 
since $r(i) > \sum_{j' < i} r(j')$ and $\ep$ is sufficiently small.
Thus, $b_i$ has the highest value of any bit, and all $w_j$ for $j \leq i$ switch to $b_i$.
When $c_i = 0$, it takes on the value of some higher bit (or zero), so $c_i$ switches to 1 and takes on the greater value of $b_i$.

\paragraph{Phase 3.}
By choice of $b_i$, we have $b_j = c_j = d_j = 1$ for all $j < i$. After the switches in Phase 2, we have $w_j = b_i$.
Thus, for any fixed $b_j$, $b_j$ follows the right column up to $d_i$, which follows its action 0 to $w_{i+1}$ and does not collect the reward $r(i)$.
Since $r(i) > \sum_{\ell < i} r(\ell)$, we have $\val(b_j) \leq \val(w_{i+1}) + \sum_{\ell < i} r(\ell) < \val(b_i) = \val(w_{j+1})$.
Thus each $b_j$ for $j < i$ switches to action 0 going to $w_{j+1}$, and Property 1 holds.
Similarly, $c_j$ and $b_j$ switch to 0, since $\val(w_{j+1}) = \val(b_i) > \val(b_j)$.
$d_i$ switches to 1, since $\val(w_{i+1}) < \val(b_i) - 2\ep$.

Observe that the invariant is again satisfied. Thus, if we start with all actions equal to 0 and each $w_i$ choosing the action taking it to the sink, the invariant is satisfied after every set of three phases.

Furthermore, since the only bit that switches to 1 is $b_i$ where $i$ is the minimum index such that $b_i = 0$, Property 2 holds.
\end{proof}

We have thus shown that hybrid policy iteration on the simple construction simulates a binary counter on the $n$ bits, taking at least $2^n$ iterations to reach the optimal policy.

\subsection{Full construction} \label{subsec:full}
We now show that we can insert a gadget at each $b_i$ node for $i \neq 1$ in order to satisfy Property 2. This gadget is depicted in \autoref{fig:b_i}.
$b_1$ remains as in the simple construction, with no added gadget. For the other nodes $b_i$, we add $f(i) + 1$ actions $ a^i_0, a^i_1, \ldots, a^i_{f(i)}$ to $b_i$, where $f(i) = 3 + 6i$. $a^i_0$ is deterministic and goes to $w_1$ with zero reward. Each other action $a^i_j$ for $j \geq 1$ returns to $b_i$ with high probability, and goes to $b_1$ and incurs a small reward with the remaining probability. The 0 actions of $b_i$ from the simple constructions are deleted; they are replaced by these new actions $a^i_0$ and $a^i_{j \geq 1}$.
These actions $a^i_j$ for $j \geq 0$ have increasing rewards (relative to index $j$) but decreasing appeals. Thus, policy iteration starts by selecting the action $a^i_0$ with the largest appeal but smallest reward and cycles through $a^i_1, a^i_2,  \ldots$ in order.
We also amend the action 1 from the simple construction to loop back to $b_i$ with overwhelming probability. Without changing any values, this lowers the appeal of action 1 so that all of the actions $a^i_j$ are more appealing.
Consequently, these extra actions delay $b_i$ from switching to action 1.

The gadget is reset ($b_i$ chooses $a^i_0$) whenever any bit other than $b_1$ is set to 1. When this bit is switched to 1, $w_1$ will have greater value than $b_1$ in Phase 3 of that switch. Thus, in Phase 3, $b_i$ will switch to the action $a^i_0$ that goes to $w_1$. 

We achieve this using the following probabilities and rewards for each action $a^i_j$:

\begin{description}
    \item[Probabilities.] $a^i_j$ goes from $b_i$ to $b_1$ with probability $\frac{1}{2^{2j}}$ and from $b_i$ back to itself with the remaining probability $1 - \frac{1}{2^{2j}}$.
    
    \item[Reward.] $a^i_j$ gives a reward of $\delta_j = 3j \cdot \delta$, where $\delta = 2^{-100n}$.
\end{description}

We also slightly modify the rewards $r(i)$, introduce costs $c(i)$, and change the actions at the nodes $w_i$ to incur these costs.
When $b_i$ is set to 0 in the full construction, its value is a bit higher than that of $w_1$ or $b_1$, which is greater than the value $b_i$ would have when set to 0 in the simple construction. Since we do not want $w_i$ to set its action to 1 until $b_i$ switches to 1, we add the cost $c(i)$ from $b_i^-$ to $b_i$ to counteract this extra value, and direct action 1 from $w_i$ to $b_i^-$. We redefine $r(i) = 2^{2i + 1}$. We set $c(i) = \frac{r(i)}{2} = 2^{2i}$ for $i > 1$, and let $c(1) = 0$.\footnote{Setting $c(1) = 0$ is for ease of notation; as there is no gadget for $b_1$, there is no real cost attached to $b_1$.} Thus, the effective reward (as seen by $w_i$) associated with $b_i = 1$ is $r(i) - c(i) = 2^{2i}$, as in the simple construction. We show later that this preserves the property that $w_i$ does not switch to 1 until $b_i$ has switched to 1.

\begin{figure}
    \centering
    \includegraphics[width=400pt]{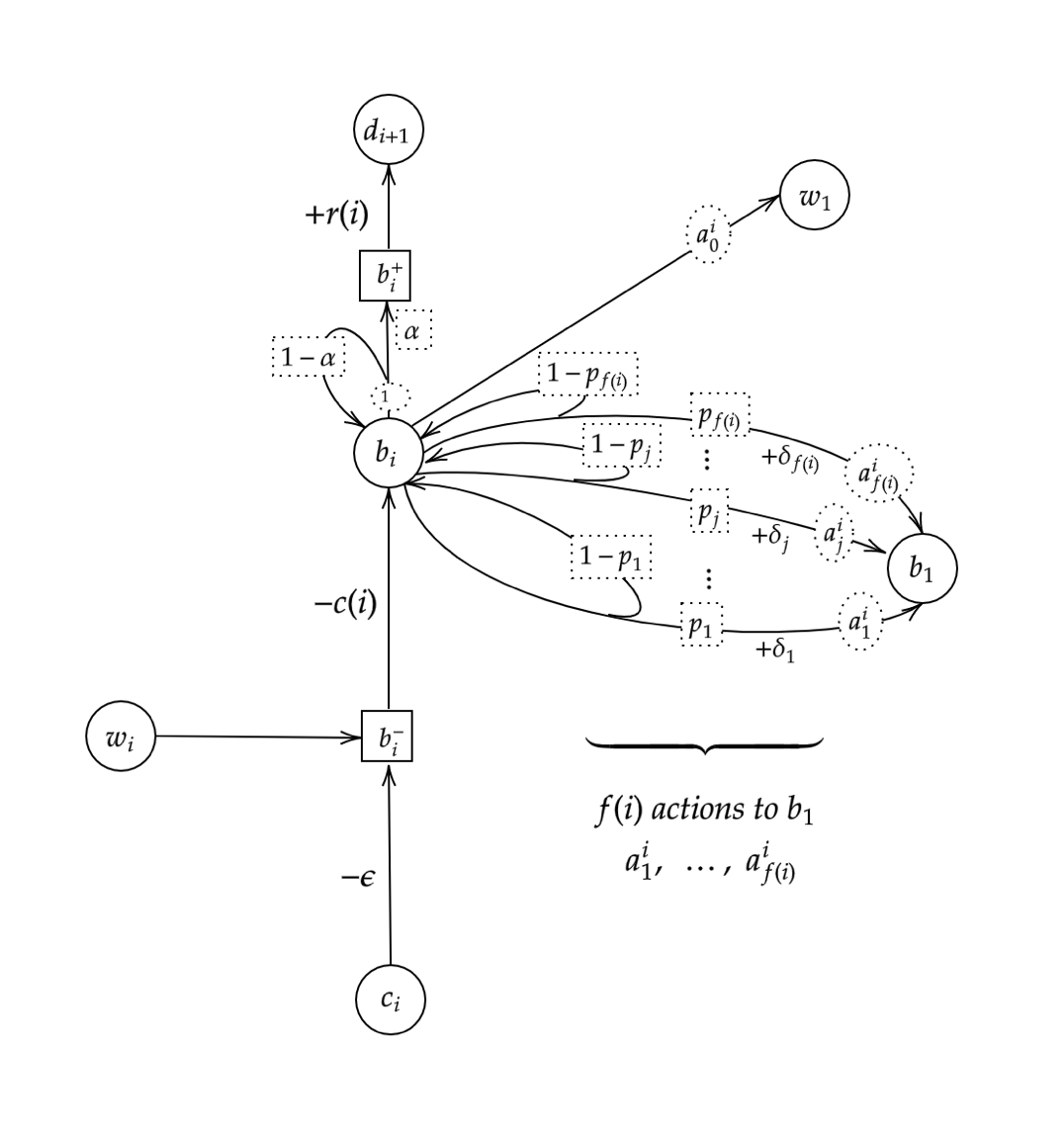}
    \caption{The structure of the full construction at node $b_i$. Probabilities are shown with dotted boxes around them; action names are shown with dotted ovals around them. Rewards are free-floating with plus or minus signs. The nodes $b_i^-$ and $b_i^+$ have only one action and are shown as squares.}
    \label{fig:b_i}
\end{figure}

We also add an action $a^i_0$ from $b_i$ to $w_1$ with no reward and with probability 1. We add a probabilistic action 1 from $b_i$ to $b_i^+$ that loops back to $b_i$ with extremely high probability $1 - \frac{1}{2^{1000n}}$.
This ensures that the appeal of action 1 is smaller than the appeal of any other action $a^i_j$.

We refer to the rewards, costs, and probabilities as \emph{parameter values}. When analyzing the behavior of policy iteration on our constructions, we point out which of our arguments depend on these exact parameter values and which do not.

\paragraph{Invariant.} At the beginning of the phases, we have $w_i = b_i = c_i = d_i = 1$ for all $i \in B$. For all $i \notin B$, we have $c_i = d_i = 0$. For all $i \notin B$ and $i > \max(B \cup \{0\})$, $w_i$ chooses its deterministic action to the sink. Otherwise, if $i \notin B$ and $i \leq \max B$, $w_i$ chooses the deterministic action to $b_\ell$ where $\ell$ is the smallest index such that $\ell \in B$ and $\ell \geq i$.
If the last bit added to $B$ was not $b_1$, $b_i = a^i_0$ for all $i \notin B$. If the last bit added to $B$ was $b_1$, we have $b_i = a^i_0$ or $b_i = a^i_3$ for all $i \notin B$. 

\paragraph{Phases.} 
Each set of 3 phases involves adding an index $i$ to $B$ and resetting all lower indices, so that $i$ is now the minimum index in $B$. Suppose first that $i \neq 1$. For each phase, we show the modified behavior compared to in the simple construction. We also introduce a new phase 0, where the bits not in $B$ cycle through their actions $a_j$.
\begin{enumerate}
    \setcounter{enumi}{-1}
    \item If $b_1$ is switchable to 1, proceed to Phase 1. Otherwise, for each $j \notin B$, $b_j$ increments its current action $a^j_m$ to the next action $a^j_{m+1}$. This repeats until $m+1 = f(i)$ for some $i$; when this happens, the next iteration begins Phase 1.
    We will show later that this $i$ is unique, and in fact $i = \min \bar{B}$.
    \item $b_i$ where $i = \min \bar{B}$ switches to 1. For all $i' \neq i$, $i' \notin B$, $b_{i'}$ switches from its current action $a^{i'}_\ell$ to $a^{i'}_{\ell+1}$.
    \item $w_j$ switches to $b_i$ for all $j \leq i$. $c_i$ switches to 1. For all $i' \neq i$, $i \notin B$, $b_{i'}$ switches from $a^{i'}_{\ell+1}$ to $a^{i'}_{\ell+2}$.
    \item $b_j$ switches to $a^j_0$ for all $j < i$ and $j > 1$. $d_i$ switches to 1. $c_j, d_j$ switch to 0 for all $j < i$. If $i > 1$, for all $i' > i$ where $i' \notin B$, $b_{i'}$ switches to $a^{i'}_0$, resetting the actions. Otherwise, if $i = 1$, for all $i' > i$, $b_{i'}$ switches from $a^{i'}_{\ell+2}$ to $a^{i'}_{\ell+3}$.
\end{enumerate}

\paragraph{Cycling, and defining $f(\cdot)$.} We let $f(i) := 3 + 6i$. We first note that this choice of $f$ is large enough that each bit $b_{i'}$ for $i' \notin B$ can indeed increment up to $a^{i'}_{\ell+3}$ when specified. More precisely, we want to show that when $b_i$ selects action $a^i_{m+1}$ in Phase 0, $b_{i'}$ is selecting an action $a^{i'}_\ell$ where $\ell \leq f(i') - 3$ (so there are enough actions for $b_{i'}$ to make its remaining increments). For each $i, i'$, $f(i)$ and $f(i')$ differ by at least 6. In the worst case, Phase 0 starts with $b_i = a^i_3$ and $b_{i'} = a^{i'}_0$ by the invariant. Thus, when $b_i = a^i_{m+1}$ in Phase 0, we have $b_{i'} = a^{i'}_\ell$ for $\ell \leq f(i') - 3$. 

We now show the properties that we claimed in Phase 0: if the invariant is satisfied at the start of Phase 0 and the bits increment their actions as described, then the unique first index $i$ to reach $m +1 = f(i)$ is $i = \min \bar{B}$.

\begin{proof}
Let the invariant be satisfied, and let $i = \min \bar{B}$, and let $i'$ be any other index not in $B$. Our proof proceeds in two cases, where (1) the previous bit added to $B$ was $b_1$, and (2) where the previous bit added to $B$ was not $b_1$.

For case (1), the invariant says that for all $i' \notin B$ (including $i' = i$), we have $b_{i'} = a^{i'}_0$ or $a^{i'}_3$.  In the worst case we start Phase 0 with $b_{i'} = a^{i'}_3$ and $b_i = a^i_0$. Thus, when $b_i$ selects an action $a^i_m$, $b_{i'}$ selects some action $a^{i'}_{m'}$ where $m' \geq m + 3$. Since $i$ and $i'$ have values of $f$ differing by at least 6, and $f$ is strictly increasing, $b_i$ reaches $m+1 = f(i)$ first.

For case (2), the invariant says that for all $i' \notin B$ (including $i' = i$), we have $b_{i'} = a^{i'}_0$. Since $f$ is strictly increasing, $i = \min \bar{B}$ also achieves the minimum value of $f(i')$ over all $i' \notin B$. 
Thus, with each $b_{i'}$ incrementing its action by one, $b_i$ will be the unique first index to achieve $m+1 = f(i)$.
\end{proof}

\subsubsection{Propositions} \label{sec:propositions}
For the sake of the following propositions, we introduce a \emph{weak invariant}. The weak invariant is the same condition as the strong invariant, except that any bit $b_i$ not in $B$ may select any action $a^i_j$.

These propositions establish relationships between the values and appeals of the vertices. They are sufficient for proving \autoref{thm:fullconstruction}. We prove them here in more generality than is necessary for \autoref{thm:fullconstruction}, since this generality will be useful when we reuse them later for the robust construction. We let $\emax$ denote the maximum value of any small cost $\ep$; here, $\emax = \ep$ since all small costs are the same.
We let $\dmax$ denote the maximum value of any small reward $\delta$; here, $\dmax = 3(6n + 3)\delta$.
We let $\ep(d_j)$ and $\ep(c_j)$ denote the values of the small $\ep$ costs on the 1 actions from $d_j$ and $c_j$ respectively. In the full construction, $\ep(d_j) = \ep(c_j) = \ep$.

In proving the propositions, we use a notion of reachability in the MDP. At any point $t$ in time, we consider the subgraph $G(t)$ induced by the actions selected at the nodes. A node $v$ is reachable from a node $u$ if there is a path from $u$ to $v$ in $G(t)$.

Propositions \ref{prop:rewcost} and \ref{prop:cost} relate the large rewards $r(i)$ and large costs $c(i)$. 
\autoref{prop:rewcost} shows that the reward $r(i)$ is substantially larger than the sum of all smaller rewards. As in the simple construction, we use this to show that when a bit $b_i$ is set to 1, all lower bits $b_j$ are enticed by its large reward and reset to 0.
\autoref{prop:cost} will be used to show that $w_i$ does not switch to any $b_j$ until $b_j$ has switched to 1.

Recall that $r(i) = 2^{2i + 1}$, and $c(i) = \frac{r(i)}{2} = 2^{2i}$. Thus, the effective reward associated with $b_i = 1$ is $r(i) - c(i) = 2^{2i}$, as in the simple construction. 

\begin{proposition} \label{prop:rewcost}
$r(i) - c(i) - 2n\emax - \dmax> \sum_{j < i} r(j)$ for all $i$.
\end{proposition}

\begin{proof}
$r(i) - c(i) = 2^{2i}$. Since $r(i) = 2^{2i+1}$, we have $\sum_{j < i} r(j) < \sum_{j = 2}^{2i - 1} 2^j < 2^{2i} - 2$.
Since $2n\emax + \dmax < 2$, $r(i) - c(i) - 2n\emax - \dmax > 2^{2i} > \sum_{j < i} r(j)$.
\end{proof}

\begin{proposition} \label{prop:cost} 
$c(i) > \delta_{\text{max}} + 2n\emax + \sum_{j < i} r(j)$ for all $i$.
\end{proposition}

\begin{proof}
In our construction, $\delta_{\text{max}} \leq 3f(n)\delta$.
Since $\delta = 2^{-100n}$, $f(n) = 3 + 6n$, and $\ep = \emax = 2^{-100n}$, $3f(n)\delta + 2n\emax < 2$. $\sum_{j < i} r(j) = \sum_{j=1}^{i-1} 2^{2j+1} < \sum_{j=2}^{2i - 1} 2^j < 2^{2i} - 2$. Thus, $\delta_{\text{max}} + 2n\emax + \sum_{j < i} r(j) < 2^{2i} = c(i)$.
\end{proof}

\autoref{prop:maxval} establishes an upper bound on the value of $b_1$, helping us later upper bound the appeal of switching any bit to 1.
\begin{proposition} \label{prop:maxval}
When the weak invariant is satisfied, $\val(b_1) \leq r(1) + \sum_{\substack{i \in B \\ i > 1}} r(i) - c(i)$. 
\end{proposition}

\begin{proof}
We first show that $b_i$ is reachable from a node other than $b_i^-$ only if $i \in B$. Assume that $i \notin B$, $i > 1$. By the weak invariant, $c_i = 0$. Similarly, since for any $j$, $w_j$ always points to the sink or a bit in $B$, no $w_j$ points to $b_i$. 
These are all possible actions to $b_i$, so $b_i$ is unreachable from any other node.

Since positive rewards are included only in the actions from the nodes $b_i$, $b_1$ attains at most the sum of the rewards incurred by the bits in $B$.
For $i > 1$, $b_i$ is reachable only from $b_i^-$; on this edge, we incur a cost of $c(i)$. 
Thus, 
$\val(b_1) \leq r(1) + \sum_{\substack{i \in B \\ i > 1}} r(i) - c(i)$.
\end{proof}

\autoref{prop:valbound} gives an upper and a lower bound for every bit $b_i$ for $i \in B$. .
\begin{proposition} \label{prop:valbound}
When the weak invariant is satisfied, for every $i \in B$ we have 
$$r(i) + \left(\sum_{\substack{j > i \\ j \in B}} r(j) - c(j)\right) - 2(n-i+1)\emax \leq \val(b_i) \leq r(i) + \left(\sum_{\substack{j > i \\ j \in B}} r(j) - c(j)\right)$$
\end{proposition}

\begin{proof}
Let $i \in B$, $i \neq 1$.
By the same argument from the proof of \autoref{prop:maxval}, there is a path from $b_i$ to $b_j$ only if $j \in B$. Rewards are only incurred on the actions from the bits. Each bit can be visited at most once, since any cycle must include an edge from some $b_j$ to $b_1$. But these actions are not reachable, since they are taken only if $j \notin B$, meaning $b_j$ is unreachable from any other node. Thus, for any $i \in B$, the value of $b_i$ is at most the sum of rewards incurred in the actions taken from higher bits in $B$, minus the costs required to reach these actions. This is exactly the upper bound,  $r(i) + \left(\sum_{\substack{j > i \\ j \in B}} r(j) - c(j)\right)$.
For $i = 1$, the upper bound holds by \autoref{prop:maxval}.

Like in the simple construction, we prove the lower bound by backwards induction on $i$.  We use the base case $i = \max B$.
If $i = n$, $\val(b_n) = r(n) - \ep(c_{n+1}) - \ep(d_{n+1})$ which is at least $r(n) - 2\emax$ and at most $r(n)$. Otherwise, if $i \neq n$, $b_i$ takes action 1 to $d_{i+1}$, collecting reward $r(i)$. $d_{i+1}$ takes action 0 to $w_{i+1}$, which goes to the sink with reward 0, since $i+1 \notin B$. Thus, $\val(b_i) = \val(d_{i+1}) + r(i) = r(i)$. 

Assume that for fixed $k \geq 0$, the claim holds for all indices $c \geq n - k$.

Let $j = n - k - 1$, where $j \in B$. By the invariant, $b_j$ selects its action 1 to $d_{j+1}$, collecting reward $r(j)$. First, suppose that $j + 1 \in B$. 
Then $d_{j+1} = c_{j+1} = 1$, so $\val(d_{j+1}) = \val(b_{j+1}) - c(j+1) - \ep(c_{j+1}) - \ep(d_{j+1}) \geq \val(b_{j+1}) - c(j+1) - 2\emax$.
By assumption, $\val(b_{j+1}) \geq r(j+1) + \left(\sum_{\substack{j' >  j+1\\ j' \in B}} r(j') - c(j')\right) - 2(n-j)\emax$. Thus, 
\begin{align*}
    \val(b_j) &\geq r(j) + r(j+1) - c(j+1) + \left(\sum_{\substack{j' >  j+1\\ j' \in B}} r(j') - c(j')\right) - 2(n-j)\emax - 2\emax\\
    &= r(j) + \left(\sum_{\substack{j' >  j\\ j' \in B}} r(j') - c(j')\right) - 2(n-j+1)\emax
\end{align*}

Next, suppose that $j + 1 \notin B$. Then $d_{j+1} = 0$, so it takes its action to $w_{j+2}$ with reward 0. $w_{j+2}$, by the invariant, goes to $b_\ell$ with reward 0, where $\ell$ is the smallest index in $B$ such that $\ell \geq j+2$. 
Since $j + 1 \notin B$, $\ell$ is also the smallest index in $B$ such that $\ell > j$.
Thus, $\val(d_{j+1}) = \val(w_{j+2}) = \val(b_\ell) - c(\ell)$, and by assumption, 
$\val(b_\ell) \geq r(\ell) + \left(\sum_{\substack{j' >  \ell \\ j' \in B}} r(j') - c(j')\right) - 2(n-\ell+1)\emax$.
Since $\val(b_j) = r(j) + \val(d_{j+1})$, putting this together, we have
\begin{align*}
    \val(b_j) &= r(j) + \val(b_\ell) - c(\ell)\\
    &\geq r(j) - c(\ell) + r(\ell) + \left(\sum_{\substack{j' >  \ell \\ j' \in B}} r(j') - c(j')\right) - 2(n-\ell+1)\emax\\
    &= r(j) + \left(\sum_{\substack{j' >  j \\ j' \in B}} r(j') - c(j')\right) - 2(n-\ell+1)\emax\\
    &\geq r(j) + \left(\sum_{\substack{j' >  j \\ j' \in B}} r(j') - c(j')\right) - 2(n-j+1)\emax
\end{align*}
as desired. We have thus shown that assuming that the lower bound holds for all indices $c \geq n - k$, the claim holds for index $j - n - k - 1$. 
\end{proof}

\autoref{prop:action1} shows that the actions $a_j^i$ for a bit $b_i$ have decreasing appeal. This property ensures that $b_i$ cycles through all of its actions $a_j^i$ before switching to 1. In its proof, we will make use of the following fact, which we state as a lemma to use it again in proving the main theorem.
\begin{lemma} \label{lemma:wb}
When the weak invariant is satisfied, regardless of the parameter values, $\val(b_1) = \val(w_1)$.
\end{lemma}
\begin{proof}
If $1 \in B$, $w_1$ goes to $b_1$ with reward 0, so their values are equal. If $1 \notin B$, $w_1$ goes to $b_\ell$ with reward 0, where $\ell$ is the smallest index in $B$ (if no such index exists, $b_\ell$ goes to the sink). 
$b_1$ goes to $w_2$ with reward 0. $w_2$ also goes to $b_\ell$ (or the sink) with reward 0. Thus, $\val(w_1) = \val(b_1)$.
\end{proof}

\begin{lemma} \label{lemma:action1}
Let the weak invariant be satisfied, where $b_i = a^i_{f(i)}$. Then for any parameters values for which Propositions 5-8 hold, the action of $b_i$ with greatest appeal is 1.
\end{lemma}

\begin{proof}
Finally, suppose that $j = f(i)$. Then $\appeal(b_i, 1) = \alpha(r(i) + \val(d_{i+1})) + (1-\alpha) \val(b_i)$, which is greater than $\val(b_i)$ if and only if $r(i) + \val(d_{i+1}) > \val(b_i)$.
We show that this is true, using the fact that $r(i)$ is greater than the sum of rewards of lower bits.

First, observe that $\val(b_i) \leq \val(b_1) + \dmax$, and by \autoref{prop:maxval}, $\val(b_1) + \dmax \leq (\sum_{\substack{j > i \\ j \in B}} r(j) - c(j)) + (\sum_{j < i} r(j)) + \dmax$.
By \autoref{prop:rewcost}, we can upper bound the sum of rewards $r(j)$ for $j < i$, and we have that $\val(b_1) + \dmax \leq (\sum_{\substack{j > i \\ j \in B}} r(j) - c(j)) + r(i) - 2n\emax $.

We now relate this expression back to $r(i) + \val(d_{i+1})$. First, consider the case where $i + 1 \notin B$. 
Here, $d_{i+1}$ goes to $w_{i+2}$ with no reward or cost. By the weak invariant, $w_{i+2}$ takes its action to $b_\ell$ and incurs cost $c(\ell)$, where $\ell$ is the smallest index such that $\ell \in B$ and $\ell \geq i$ (since $i + 1 \notin B$). Thus, $\val(d_{i+1}) = \val(b_\ell) - c(\ell)$.
If $i + 1 \in B$, $d_{i+1}$ incurs cost $\ep$ and goes to $c_{i+1}$, which incurs cost $\ep = \emax + c(i+1)$ and goes to $b_{i+1}$. Here, $\val(d_{i+1}) = \val(b_{i+1}) - c(i+1) - 2\emax$, where $i + 1$ is the smallest index that is greater than $i$ and in $B$.
In other words, because of how $\ell$ is defined, regardless of whether $i + 1 \in B$, $\val(d_{i+1}) \geq \val(b_\ell) - c(\ell) - 2\emax$.
By \autoref{prop:valbound}, $\val(b_\ell) - c(\ell) - 2\ep \geq (\sum_{\substack{j > i \\ j \in B}} r(j) - c(j)) - 2n\emax$.

Thus, $\val(d_{i+1}) + r(i) > (\sum_{\substack{j > i \\ j \in B}} r(j) - c(j)) + r(i) - 2n\emax \geq \val(b_1) + \dmax$, meaning that the appeal of switching $b_i$ to 1 is greater than the current value of $b_i$. In other words, if $b_i = a^i_{f(i)}$, $b_i$ is switchable to action 1.
Action 1 has the greatest appeal of any action, since the other actions $a^i_{j'}$ for $j' < j$ have reward at most $\dmax$ and thus appeal at most $\val(b_1) + \dmax$.
\end{proof}

\begin{proposition} \label{prop:action1}
If the weak invariant is satisfied, and $b_i = a^i_j$, $b_i$ is switchable and the action with greatest appeal is $a^i_{j+1}$ if $j \neq f(i)$, and action 1 if $j = f(i)$.
\end{proposition}

\begin{proof}
We first show that action 1 has low appeal. Action 1 loops back to $b_i$ with very high probability $\alpha = 1 - \frac{1}{2^{1000n}}$, so $\appeal(b_i, 1) \leq \val(b_i) + \frac{1}{2^{1000n}} (r(i) + \val(d_{i+1}))$. By the same argument as in the proof of \autoref{prop:maxval}, the only rewards reachable from $d_{i+1}$ are those of the bits in $B$.
These rewards $r(j)$ are collected at most once, and they are collected with their corresponding costs $c(j)$ with the exception of $r(1)$ which has no cost. Thus, 
$$r(i) + \val(d_{i+1}) \leq r(i) + r(1) + \sum_{j = 2}^n r(j) - c(j) \leq 2^{2i + 1} + 2^5 + \sum_{j=2}^n 2^{2j} \leq 2^{2n+1}$$
So $\appeal(b_i, 1) \leq \val(b_i) + \frac{1}{2^{1000n}} \cdot 2^{2n+1} \leq \val(b_i) + 2^{-998n + 1}$.

Before arguing about the relative appeals of the actions $a^i_{j'}$, we note that action $a^i_0$ goes to $w_1$ with reward 0, and $\appeal(b_i, a^i_0) = \val(w_1) = \val(b_1)$ when the weak invariant is satisfied by \autoref{lemma:wb}. It thus holds that for all $j$, $\appeal(b_i, a^i_j) = \val(b_1) + 3j\cdot \delta$.

If an action $a^i_j$ (including $j = 0$) is currently selected where $j \neq f(i)$, $\val(b_i) = \val(b_1) + 3j \cdot \delta$, so the appeal of action $a^i_{j + c}$ is $(1 - \frac{1}{2^{2(j+c)}})(\val(b_1) + 3j \cdot \delta) + \frac{1}{2^{2(j+c)}}(\val(b_1) + 3(j + c)\delta) = \val(b_i) + \frac{3c\delta}{2^{2(j+c)}}$.
This expression is maximized for $c = 1$. $\appeal(b_i, a^i_{j+1}) > \val(b_i)$, so $b_i$ is switchable, and action $a^i_{j+1}$ is more appealing than any other action $a^i_{j+c}$ for $c > 1$. Actions $a^i_{j-c}$ for $c > 0$ have appeal less than $\val(b_i)$, since they all go to $b_1$ (or $w_1$, whose value is equal to that of $b_1$ by \autoref{lemma:wb}) with reward less than the reward of $a^i_j$. Furthermore, since $\delta = 2^{-100n}$, we have $2^{-998n +1} < \frac{3\delta}{2^{2n}}$ and thus $\appeal(b_i, 1) < \appeal(b_i, a^i_{j+1})$.

Finally, by \autoref{lemma:action1}, if $b_i = a^i_{f(i)}$, the action with greatest appeal is 1.
\end{proof}

\subsubsection{Main theorem}
Our main result is that on the full construction, Greedy PI takes subexponentially many iterations. The proof follows the same structure as the proof of the analogous result for the simple construction. 
We first prove \autoref{lemma:switchable}, which is analogous to \autoref{lemma:phases}. 
We then prove \autoref{thm:fullconstruction}, arguing as in the simple construction that Greedy PI follows our prescribed phases.
This shows that with the all-zero starting policy, Greedy PI behaves like a binary counter, iterating through all $2^n$ bit strings to reach the optimal policy.

\paragraph{All-zero policy.}The all-zero policy $\pi_0$ for the full construction is the same as that of the simple construction for all nodes other than the $b_i$ nodes. Each node $b_i$ selects its action $a^i_0$.

\begin{lemma} \label{lemma:switchable}
Let Propositions 5 through 9 hold given the parameter values. When the weak invariant is satisfied, the set of switchable nodes is exactly the set of bits not in $B$.
\end{lemma}

\begin{proof}
We first reiterate that by \autoref{lemma:wb}, which holds regardless of the parameter values, $\val(w_1) = \val(b_1)$ when the weak invariant is satisfied.

\paragraph{Nodes $b_i$ for $i \notin B$:} All such nodes $b_i$ are switchable by \autoref{prop:action1}. 

\paragraph{Nodes $b_i$ for $i \in B$:} Let $i \in B$. Then by \autoref{prop:valbound}, $\val(b_i) \geq r(i) + (\sum_{\substack{j' >  i \\ j' \in B}} r(j') - c(j')) - 2n\emax$. If $b_i$ switches to some action $a^i_j$, it will take on at most $\val(w_1) + \dmax = \val(b_1) + \dmax$. By \autoref{prop:maxval}, $\val(w_1) = \val(b_1) \leq r(1) + \sum_{\substack{j' \in B \\ j' > 1}} r(j') - c(j') \leq \val(b_i) - c(i) + 2n\emax + \sum_{\substack{j' \in B \\ j' < i}} r(j')$. Since $c(i) > 2n\emax + \dmax + \sum_{\substack{j' \in B \\ j' < i}} r(j')$ by \autoref{prop:cost}, $\val(b_i) \geq \appeal(b_i, a^i_j)$ for any $j$.

\paragraph{Nodes $c_i, d_i$:} Let $i \in B$. Then $c_i = d_i = 1$, and $\val(c_i) = \val(b_i) - c(i) - \ep(c_i) > 0$, $\val(d_i) = \val(b_i) - c(i) - \ep(c_i) - \ep(d_i) > 0$, where the fact that these values are positive comes from \autoref{prop:rewcost}.
If $c_i$ or $d_i$ were to switch to 0, their values would be $\val(w_{i+1}) = \val(b_\ell) - c(\ell)$, where $\ell$ is the next highest bit in $B$, or 0 if there is no higher bit in $B$. If there is no higher bit in $B$ and the value of $w_{i+1}$ is 0, $c_i$ and $d_i$ are not switchable since their current values are positive. 
If there is such a $b_\ell$, then by \autoref{prop:valbound}, $\val(b_\ell) - c(\ell) \leq \sum_{\substack{j \geq \ell \\ j \in B}} r(j) - c(j)$, so $\val(w_{i+1}) \leq \sum_{\substack{j > i \\ j \in B}} r(j) - c(j)$. 
But also by \autoref{prop:valbound}, $\val(b_i) \geq r(i) + (\sum_{\substack{j > i \\ j \in B}} r(j) - c(j)) - 2(n-i+1)\emax$. So $\val(b_i) - \val(w_{i+1}) \geq r(i) - 2(n-1+1)\emax$. By \autoref{prop:rewcost}, this difference is positive.

Let $i \notin B$. Then $\val(c_i) = \val(d_i)$, so $d_i$ is not switchable, since switching to 1 incurs $d_i$ a cost of $\ep(d_i)$. For $c_i$, observe that $\val(b_i) \leq \val(b_1) + \dmax$ or $\val(b_i) \leq \val(w_1) + \dmax$, depending on which $a^i$ action it chooses. As previously noted, $\val(w_1) = \val(b_1)$.
Because of the cost $c(i)$ between $c_i$ and $b_i$, the appeal of $c_i$ switching to 1 is at most $\val(b_1) + \dmax - c(i)$.
By \autoref{prop:maxval}, $\val(b_1) \leq r(1) + \sum_{\substack{j \in B \\ j > 1}} r(j) - c(j)$, so plugging this in we have $\appeal(c_i, 1) \leq \dmax + r(1) + (\sum_{\substack{j \in B \\ j > 1}} r(j) - c(j)) - c(i)$.
By \autoref{prop:cost}, $c(i) > \dmax + 2n\emax + \sum_{j < i} r(j)$, which we can subtract from the sum in the previous inequality to get $\appeal(c_i, 1) < (\sum_{\substack{j \in B \\ j \geq i}} r(j) - c(j)) - 2n\emax$.
Now, we relate this appeal back to the current value of $c_i$. If there are no bits in $B$ that are higher than $i$, this value is negative, and $c_i$ is not switchable to 1 as its starting value is positive.
Otherwise, let $\ell$ again be the next highest bit in $B$.
 Recall that $c_i$ goes to $w_{i+1}$ with reward 0, and $\val(c_i) = \val(w_{i+1}) = \val(b_\ell) - c(\ell) \geq (\sum_{\substack{j \in B \\ j \geq \ell}}r(j) - c(j) ) - 2n\emax$ by \autoref{prop:valbound}.
 Since $\ell$ is the next highest bit (after $i$) in $B$, the two sums are equal: $(\sum_{\substack{j \in B \\ j \geq \ell}}r(j) - c(j) ) = (\sum_{\substack{j \in B \\ j \geq i}}r(j) - c(j) )$. Thus, our strict upper bound on $\appeal(c_i, 1)$ exactly equals this lower bound on $\val(c_i)$, making $c_i$ not switchable to 1.

\paragraph{Nodes $w_i$:} $w_i$ points to $b_\ell^-$, where $\ell$ is the smallest index such that $\ell \in B$, $\ell \geq i$.
Thus, $\val(w_i) = \val(b_\ell) - c(\ell)$. 
Recall that by \autoref{prop:valbound}, $\val(b_\ell) \geq r(\ell) + \left(\sum_{\substack{j' > \ell \\ j' \in B}} r(j') - c(j')\right) - 2n\emax$, so $ \val(b_\ell) -  c(\ell) \geq (\sum_{\substack{j' \geq \ell \\ j' \in B}} r(j') - c(j')) - 2n\emax$. 

We first show that for any $j > i$, $j \neq \ell$, $w_i$ is not switchable to $b_j$. 
If $j \in B$, by \autoref{prop:valbound}, $\appeal(w_i, b_j)  = \val(b_j) - c(j) \leq (\sum_{\substack{j' \geq j \\ j' \in B}} r(j') - c(j')) \leq (\sum_{\substack{j' \geq \ell \\ j' \in B}} r(j') - c(j')) + r(\ell)-c(\ell) - 2n\emax $, where the last inequality follows from the fact that $\ell < j$ is missing from the summation, 
and $r(\ell) - c(\ell) - 2n\emax > 0$ by \autoref{prop:rewcost}.
Thus, $\val(w_i) = \val(b_\ell) - c(\ell) > \appeal(w_i, b_j)$, and $w_i$ is not switchable to $b_j$. 

If $j \notin B$, $\val(b_j) \leq \val(b_1) + \dmax$. The same argument for why $c_i$ is not switchable to 1 when $i \notin B$ holds here. 

Finally, $w_i$ is never switchable to the sink, since the value of the sink is zero, and the value of $w_i$ is never negative.

\end{proof}

\begin{theorem} \label{thm:fullconstruction}
Let the parameter values be such that propositions 5 through 9 hold.
If we start at the all-zero policy, Greedy PI on the full construction takes at least $2^n$ iterations to arrive at the optimal policy under the total reward criterion. 
\end{theorem}

\begin{proof}
The outline of this proof is similar to that of the proof of \autoref{thm:hybrid}, the equivalent theorem for the simple construction. 
We start with the all-zero policy, which satisfies the invariant.
We first show that when the invariant is satisfied, Greedy PI proceeds in the four previously described phases, ending Phase 4 with the invariant again satisfied. 
Since each set of four phases involves adding the lowest zero bit to $B$ and resetting all lower bits to zero, Greedy PI behaves exactly as a binary counter and iterates through all binary strings for the bits $b_i$ before reaching the optimal all-one policy, where all bits $b_i$ are set to 1.
Thus, it remains only to show that Greedy PI follows the four phases, preserving the invariant at the end of Phase 4.

Suppose that the invariant holds. We show that Greedy PI proceeds in our described phases.
\paragraph{Phase 0.} The weak invariant is maintained throughout this phase. By \autoref{lemma:switchable}, the set of switchable nodes is exactly the set of $b_j$ where $j \notin B$. If no $b_j$ already selects $a^j_{f(j)}$, by \autoref{prop:action1}, each $b_j$ increments its action $a^j_\ell$ to $a^j_{\ell+1}$.
Otherwise, some $b_i$ selects $a^i_{f(i)}$, and Phase 0 has completed. There can only be one such $b_j$, because the invariant requires that every node $b_{j'}$ for $j' \notin B$ starts by selecting $a^{j'}_0$ or $a^{j'}_3$, and $f(\cdot)$ differs by at least 6 on any pair of distinct inputs.

\paragraph{Phase 1.} Because the weak invariant is still satisfied, the set of switchable nodes is exactly the set of $b_j$ where $j \notin B$. $b_i$, which currently selects $a^i_{f(i)}$, switches to action 1.
By \autoref{prop:action1}, each $b_j$ where $j \neq i$, $j \notin B$ increments its action. 

\paragraph{Phase 2.} Since no node selects an action to $b_i$, changing the value of $b_i$ can only make $w_j$ for $j \leq i$ and $c_i$ switchable. Furthermore, the values of all nodes other than $b_i$ are the same as when the weak invariant is satisfied.
Thus, $\val(b_i) = r(i) + \val(d_{i+1})$. 

We first show that $w_j$ switches to $b_i$ for all $j \leq i$. For any such $w_j$, $\appeal(w_j, b_i) = \val(b_i) - c(i) \geq (\sum_{\substack{j' \geq i \\ j' \in B}} r(j') - c(j')) -2n\emax$.
For any $\ell$ such that $\ell \neq i$, we have $\appeal(w_j, b_\ell) \leq \sum_{\substack{j' > i \\ j' \in B}} r(j') - c(j') + \sum_{\substack{j' < i \\ j' \in B}} r(j')$. Since $r(i) - c(i) - 2n\emax > \sum_{\substack{j' < i \\ j' \in B}} r(j')$ by \autoref{prop:rewcost}, $\appeal(w_j, b_i) > \appeal(w_j, b_\ell)$. 
Thus, $w_j$ switches to $b_i$ for all $j \leq i$.

The only other node that may become switchable is $c_i$. $c_i$ currently selects action 0 to $w_{i+1}$ with zero reward, so $\val(c_i)  = \appeal(c_i, 0) = \val(w_{i+1})$. Taking action 1 has $\appeal(c_i, 1) \geq \val(b_i) - \emax = \val(d_{i+1}) + r(i) - \emax$. 
We claim that the values of $d_{i+1}$ and $w_{i+1}$ are nearly equal. If $i+1 \in B$, $d_{i+1}$ takes its action 1 to $c_{i+1}$, which takes its action 1 to $b_i^-$. A small cost of at most $\emax$ is incurred on each of these edges.
$w_{i+1}$ takes its action 1 with no cost, also leading to $b_i^-$.
Thus, $\appeal(c_i, 1) \geq \val(d_{i+1}) + r(i) - 3\emax > \appeal(c_i, 0)$.
If $i + 1 \notin B$, $d_{i+1}$ takes its action 0 to $w_{i+2}$, which takes its action with zero reward to the least bit in $B$ that is at least $i+2$ (or the sink, if no such bit exists). $w_{i+1}$ takes its action to that same bit, also with zero reward.
Therefore, $\val(d_{i+1}) = \val(w_{i+1})$, and $\appeal(c_i, 1) > \val(d_{i+1}) = \appeal(c_i, 0)$. We've shown that in either case, $c_i$ is switchable to 1. 

Finally, we note that the nodes $b_{i'}$ where $i' \neq i$ and $i' \in B$ remain switchable. Since the values of all nodes reachable by the $b_{i'}$'s remain unchanged, these nodes again increment their actions.

\paragraph{Phase 3.} We first show a useful fact, that $\val(w_1) = \val(b_i) - c(i) > \val(b_j) - c(j) + 2n\emax + \dmax$ for any $j \neq i$. This will later help us show that the lower bits switch to the action $a^j_0$, going to $w_1$ with reward 0.

$w_1$ switched to point to $b_i^-$ in Phase 2, so its value is now $\val(b_i) - c(i)$. Since the weak invariant is still satisfied for all higher indexed nodes (than $i$), we still have $\val(b_i) \geq r(i) + \left(\sum_{\substack{j > i \\ j \in B}} r(j) - c(j)\right)$ by \autoref{prop:valbound}.
\autoref{prop:valbound} already gives us that for any $j > i$ where $j \in B$, $\val(b_j) - c(j) + 2n\emax + \dmax \leq \left(\sum_{\substack{i' > i \\ i' \in B}} r(j) - c(j)\right) + 2n\emax + \dmax < \val(b_i) - c(i)$ as desired, since $r(i) - c(i) > 2n\emax + \dmax$ by \autoref{prop:rewcost}. If $j > i$ but $j \notin B$, $b_j$ goes to $w_1$ with reward at most $\dmax$; thus, $\val(b_j) - c(j) + 2n\emax + \dmax < \val(w_1)$.
We now consider $j < i$. For all $j' < i$, recall that $b_{j'} = c_{j'} = d_{j'} = 1$ by the invariant. Thus, the nodes $d_j, c_j, b_j$ for $j < i$ are all set to 1 and form a path up the right side of the graph to $d_i$, as in the simple construction. 
Thus, $\val(b_j) \leq \val(d_i) + \sum_{\substack{i > j' > j \\ j' \in B}} r(j') - c(j')$.
Since $d_i = 0$, $\val(d_i) \leq \val(w_{i+1}) \leq \left(\sum_{\substack{j' > i \\ j' \in B}} r(j') - c(j')\right)$. 
Again, \autoref{prop:valbound} still holds for $b_i$, so $\val(b_i) \geq r(i) + (\sum_{\substack{j' \geq i \\ j' \in B}} r(j') - c(j')) - 2n\emax$. Since $r(i) - c(i) - 2n\emax - \dmax > \sum_{j' < i} r(j')$ by \autoref{prop:rewcost}, $\val(w_1) = \val(b_i) - c(i) > \val(d_i) + 2n\emax - \dmax \geq \val(b_j) + 2n\emax + \dmax$.
Thus, $\val(w_1) > \val(b_j) + 2n\emax + \dmax$ for any $j \neq i$.

We first show that $b_j$ switches to $a^j_0$ for all $j < i$. $b_j$ is currently set to 1 by the invariant. Since $\val(w_1) > \val(b_j) - c(j)$, $b_j$ is switchable to $a^j_0$. All other actions $a^j_c$ for $c > 0$ go to $b_1$ with negligible reward at most $\dmax$. Thus, if $i \neq 1$, $\val(w_1) > \appeal(b_j, a^j_c)$ for all $c > 0$.
If $i = 1$, the actions' appeals are decreasing by their indices. Thus, $b_j$ increments its action.

We next show that $d_i$ switches to 1. We showed at the beginning of this phase's proof that $\val(b_i) - c(i) > \val(b_j) - c(j) + 2n\emax$ for any $j \neq i$. 
Since $w_{i+1}$ points to some higher bit $j \neq i$, $\val(w_{i+1}) \leq \val(b_j) - c(j)$ for that $j$.
Thus, $\appeal(d_i, 1) = \val(b_i) - c(i) - \ep(d_i) > \val(w_{i+1})$.

We next show that $c_j$ and $d_j$ switch to 0 for all $j < i$. We have already shown that $\val(b_i) - c(i) = \val(w_1) > \val(b_1) + \dmax$. Since $w_j$ switched to $b_i$ in Phase 2, we have that $\appeal(c_j, 0) = \appeal(d_j, 0) = \val(w_j) = \val(b_i) - c(i)$. The current value of $c_j$ and $d_j$ is at most $\val(b_j) \leq \val(b_1) + \dmax$. We have already shown that $\val(b_i) - c(i) > \val(b_1)  + \dmax$. Thus, $c_j$ and $d_j$ switch to 0 and go to $w_j$, whose value is $\val(b_i) - c(i)$.

For all $i' > i$, $i' \notin B$, $b_{i'}$ switches to $a^{i'}_0$ if $i \neq 1$. $\appeal(b_{i'}, a^{i'}_0) = \val(w_1)$. As previously shown, $\val(w_1) > \val(b_{j}) + \dmax$ for all $j \neq i$. Thus, $b_{i'}$ is switchable to $a^{i'}_0$. 
$b_{i'}$ is not switchable to 1, since its resulting value would be at most $\val(b_j) - c(j)$ for some higher $j$, and $\val(b_j) - c(j) < \val(w_1)$.
Furthermore, for any $c > 0$, $\appeal(b_{i'}, a^{i'}_c) \leq \val(b_1) + \dmax < \val(w_1)$. Thus, $a^{i'}_0$ is the action with highest appeal, and $b_{i'}$ switches to $a^{i'}_0$.
If $i = 1$, we again have $\val(w_1) > \val(b_{j})$ for all $j \neq i$, so $b_{i'}$ is switchable. $b_{i'}$ is not switchable to 1 for the same reason as before. Since $w_1 = 1$, $\val(w_1) = \val(b_1)$, and the appeals of the actions at $b_{i'}$ are ordered based on their indices. 
Thus, $b_{i'}$ increments its action.

\end{proof}

We showed previously that for our parameter values ($r(i) = 2^{2i + 1}, c(i) = 2^{2i}, \epsilon = 2^{-100n}$, and $\delta = 2^{-100n}$), the propositions hold. Therefore, Greedy PI takes at least $2^n$ iterations to arrive at the optimal policy on the full construction MDP with these parameter values.

This result also holds for the average reward criterion. Policy iteration under the average reward criterion determines which actions to switch based first on a \emph{gain function}, then based on a \emph{bias function} in the case that multiple actions yield equal gain. 
Fearnley \cite{fearnley2010exponential} notes that for MDPs that are guaranteed to reach a 0-reward sink state, like our full construction, this gain function is always zero. 
Here, the bias function also becomes equivalent to our appeal function.
Thus, on the full construction, the choices that Greedy PI makes under the average reward criterion are the same as those made under the total reward criterion. This gives us the following corollary:

\begin{corollary}
Let the parameter values be such that propositions 5 through 9 hold.
If we start at the all-zero policy, Greedy PI on the full construction takes at least $2^n$ iterations to arrive at the optimal policy under the average reward criterion. 
\end{corollary}

\subsection{Robust construction} 

We show that our construction can be made robust to perturbations of the rewards and probabilities by replacing certain edges with gadgets.
We first show a gadget that allows us to manufacture a reward between $2^i$ and $2^{i+2}$ with probability 1.

\begin{figure}
    \centering
    \includegraphics[width=380pt]{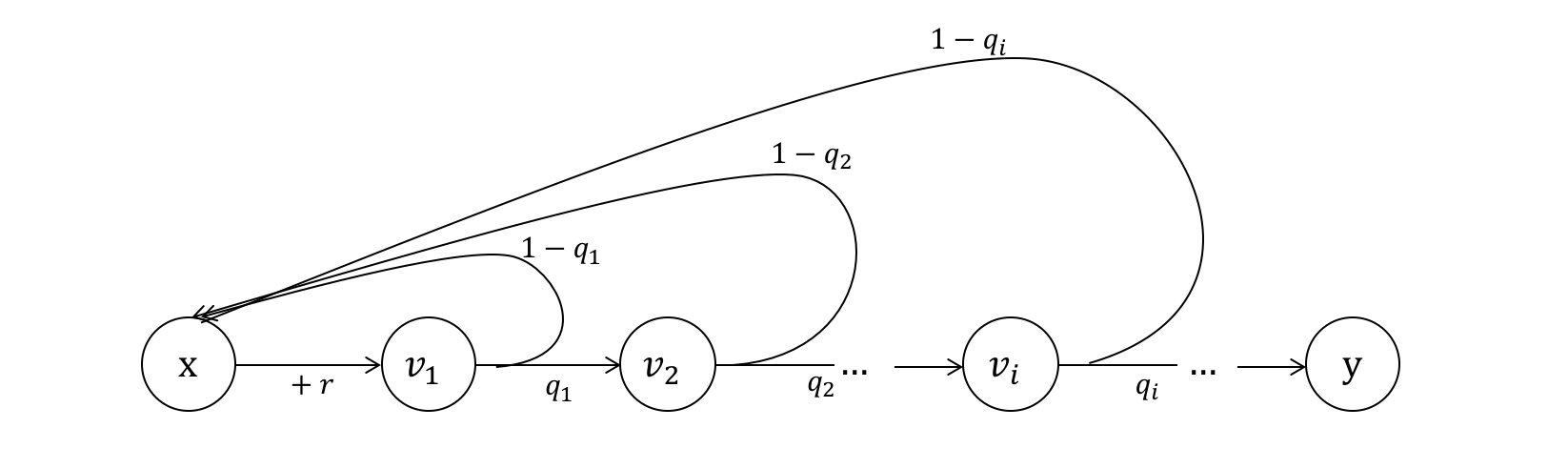}
    \caption{Gadget $g_2(k)$ with $k$ intermediate vertices allows us to manufacture an exponentially large reward between $x$ and $y$.}
    \label{fig:rewgadget}
\end{figure}

\begin{lemma}
If the probabilities and rewards are perturbed by at most $\frac{1}{4k'}$ for any $k' \geq k$, the gadget $g_2(k)$ with reward $r = 1 - \frac{1}{2k}$ and probabilities $q_i = \frac{1}{2} + \frac{1}{4k'}$ yields $\val(y) + 2^{k-2} \leq \val(x) \leq \val(y) + 2^{k}$ after perturbation.
\end{lemma}

\begin{proof}
First, we show that without perturbation, we have $\val(x) = \frac{r}{\prod_{i=1}^k q_i} + \val(y)$.
\begin{align*}
    \val(x) &= r + \val(v_1) \\
    &= r + \left( \prod_{i=1}^k q_i\right) \val(y) + \left(1-\prod_{i=1}^k q_i \right) \val(x)\\
    \val(x) &= \frac{r}{\prod_{i=1}^k q_i} + \val(y)
\end{align*}
Let $q'_i$ denote the value of $q_i$ after perturbation.
If we choose $q_i = \frac{1}{2} + \frac{1}{4k'}$, then after perturbation we must have $q'_i \in \left[\frac{1}{2}, \frac{1}{2} + \frac{1}{2k'}\right]$ for each $i$, we have $\frac{1}{2^k} \leq \prod_{i=1}^k q'_i \leq \frac{e}{2^k}$. 
Let $r'$ denote the value of the reward $r$ after perturbation.
If we choose $r = 1 - \frac{1}{2k}$, then we must have $r' \in [1 - \frac{1}{k}, 1]$.
Observe that $r'(\prod_{i=1}^k q'_i)^{-1} \geq \frac{2^k(1-\frac{1}{k})}{e} \geq 2^{k-2}$ for $k > 3$.
We also have $r'(\prod_{i=1}^k q'_i)^{-1} \leq 2^k$. 
\end{proof}
Note that we can similarly construct a negative reward (cost) between $-2^{i+2}$ and $-2^i$ by using $r = - (1-\frac{1}{2k})$.

We now show that we can use the gadget $g_2$ with zero reward to create exponentially small probabilities.

\begin{lemma} \label{lemma:g2prob}
Let $\appeal(x)$ be the appeal of taking the action leading to the gadget $g_2$ from $x$.
If the probabilities and rewards are perturbed by at most $\frac{1}{3000n^2}$, then for any $k \leq 1000n + 2$, the gadget $g_2(k)$ with reward $r = 0$ and probabilities $q_i = \frac{1}{n} + \frac{1}{1000n^2}$ yields $\appeal(x) = (1 - p)\val(x) + p\val(y)$ after perturbation, for $\frac{1}{n^k} \leq p \leq \frac{1}{n^{k-2}}$.
\end{lemma}

\begin{proof}
First, observe that
\begin{align*}
    \appeal(x) &= \left(\prod_{i=1}^k q_i\right) \val(y) + \left(1 - \prod_{i=1}^k q_i \right) \val(x)
\end{align*}
where each $q_i$ satisfies $\frac{1}{n} \leq q_i \leq \frac{1}{n} + \frac{1}{2000n^2}$. 
Thus, $\frac{1}{n^k} \leq \prod_{i=1}^k q_i \leq \frac{1}{n^k}(1 + \frac{1}{2000n})^k$. For the right side of the inequality, we have that $(1 + \frac{1}{2000n})^k \leq e$, since $k \leq 1000n + 2$. Since $e \leq n^2$, we have $(1 + \frac{1}{2000n})^k \leq \frac{1}{n^{k-2}}$, as desired.

\end{proof}

We now show a gadget allowing us to make exponentially small rewards. This is used to make the $\ep$ costs in the full construction.

\begin{figure}
    \centering
    \includegraphics[width=400pt]{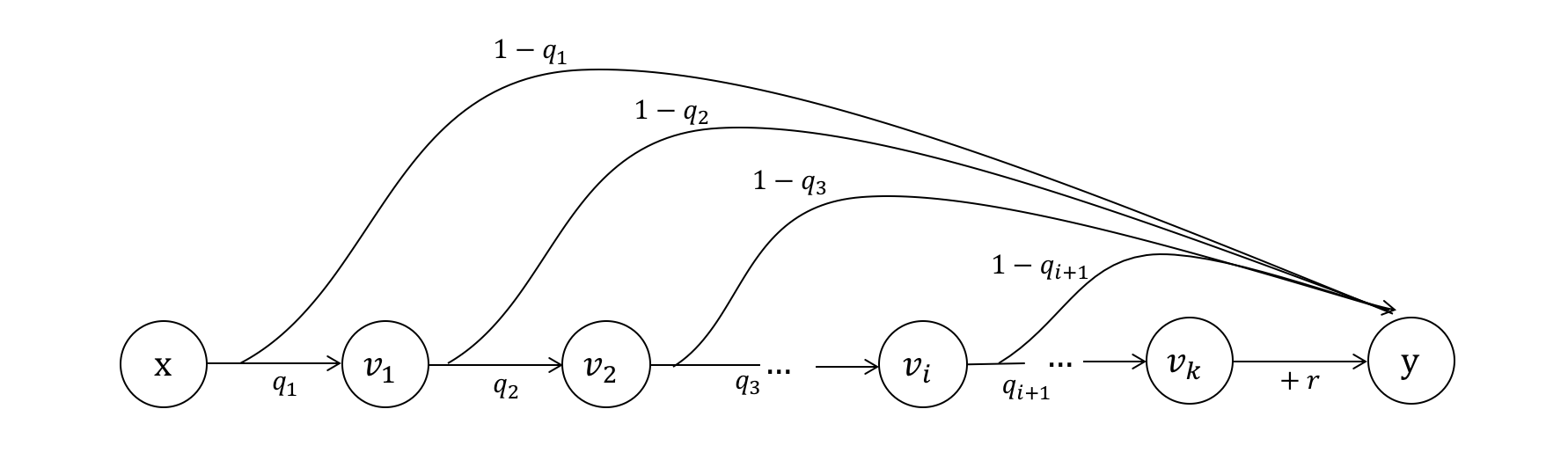}
    \caption{Gadget $g_3(k)$ with $k$ intermediate vertices allows us to manufacture an exponentially small reward between $x$ and $y$.}
    \label{fig:gadget_small_reward}
\end{figure}

\begin{lemma}
 If the probabilities and rewards are perturbed by at most $\frac{1}{4k'}$ where $k' \geq k$, the gadget $g_3(k)$ with reward $r = 1 + \frac{1}{4k}$ and probabilities $q_i = \frac{1}{2} + \frac{1}{4k'}$ yields $\val(y) + 2^{-k} \leq \val(x) \leq \val(y) + 2^{-k + 2}$ after perturbation. 
\end{lemma}

\begin{proof}
 Observe that the probability of traversing the edge with the reward $r$ is $\prod_{i=1}^k q_i$. Thus, without perturbations, $\val(x) = \val(y) + r\prod_{i=1}^k q_i$.
For each $i$, let $q'_i$ denote the probability $q_i$ after perturbation. Choosing $q_i = \frac{1}{2} + \frac{1}{4k'}$ for each $i$, we have $q'_i \in [\frac{1}{2}, \frac{1}{2} + \frac{1}{2k'}]$.
Thus, $\frac{1}{2^k} \leq \prod_{i=1}^k q'_i \leq \frac{e}{2^k}$.
Let $r'$ denote the reward $r$ after perturbation. Choosing $r = 1 + \frac{1}{4k}$, we have $r' \in [1, 1 + \frac{1}{2k}]$. Thus, $\frac{1}{2^k} \leq r' \prod_{i=1}^k q'_i \leq \frac{e(1 + \frac{1}{2k})}{2^k} \leq \frac{1}{2^{k-2}}$ for $k > 2$.
\end{proof}
Again, we can similarly construct a small negative reward between $-2^{-k+2}$ and $-2^{-k}$ by using $r = -(1 + \frac{1}{4k})$.

\subsubsection{Constructing the parameters $r(i), c(i), \ep, \delta_j, p_j, \alpha$} \label{sec:robustparams}

We first describe how to construct each of these parameters, either naively or by using the gadgets. Note that the \emph{raw} values of the rewards and probabilities used in the MDP, presented below, lie in $[-2, 2]$. While we later discuss the \emph{effective} rewards  yielded by the gadgets, which are exponential, the actual parameters are small.

\begin{description}
\item[\textbf{Reward $r(i)$}.] We construct $r(i)$ using $g_2(7(i-1) + 6)$, with $k' = 10n$. Each probability $q_j$ is set to $\frac{1}{2} + \frac{1}{40n}$, with allowable perturbation up to $\frac{1}{40n}$. The reward $r$ used in the gadget is set to $1 - \frac{1}{2(7(i-1) + 6)}$, with perturbation up to $\frac{1}{40n}$.

\item[\textbf{Cost $c(i)$}.] We construct $c(i)$ using $g_2(7(i-1) + 3)$, with $k' = 10n$. Each probability $q_j$ is set to to $\frac{1}{2} + \frac{1}{40n}$, with allowable perturbation up to $\frac{1}{40n}$. The cost $r$ used in the gadget is set to $- (1 - \frac{1}{2(7(i-1) + 3)})$, with perturbation up to $\frac{1}{40n}$.

\item[\textbf{Cost $\ep$}.] We construct $\ep$ using $g_3(100n)$, with $k' = 100n$. Each probability $q_j$ is set to $\frac{1}{2} + \frac{1}{400n}$, with allowable perturbation up to $\frac{1}{400n}$. The cost $r$ used in the gadget is set to $- (1 + \frac{1}{200n})$, with perturbation up to $\frac{1}{400n}$.

\item[\textbf{Reward $\delta_j$}.] We set $\delta_j = \frac{2j}{4n^2} + \frac{1}{8n^2}$ and allow perturbation up to $\frac{1}{8n^2}$.

\item[\textbf{Probability $p_j$}.] We construct $p_j$ using $g_2(4j+3)$ as described in \autoref{lemma:g2prob}, with reward $r = 0$, probabilities $q_\ell$ is set to $\frac{1}{n} + \frac{1}{1000n^2}$, with allowable perturbation up to $\frac{1}{2000n^2}$. \autoref{lemma:g2prob} applies because $4j + 3 \leq 1000n + 2$.

\item[\textbf{Probability $\alpha$}.] Recall that $\alpha$ is the probability associated with action 1 from each $b_i$. We construct $\alpha$ using $g_2(1000n + 2)$ with reward $r = 0$, as described in \autoref{lemma:g2prob}. Each probability $q_j$ within the gadget is set to $\frac{1}{n} + \frac{1}{1000n^2}$, with allowable perturbation up to $\frac{1}{2000n^2}$. 
\end{description}

Below, we present the \emph{effective} rewards and probabilities; i.e., the intervals that the parameters lie in after perturbation:
\begin{itemize}
    \item $r(i) \in [2^{7(i-1) + 4}, 2^{7(i-1) + 6}]$
    \item $c(i) \in [2^{7(i-1) + 1}, 2^{7(i-1) + 3}]$
    \item $\ep \in [2^{-100n}, 2^{-100n + 2}]$
    \item $\delta_j \in [\frac{2j}{4n^2}, \frac{2j+1}{4n^2}]$.
    $0 < \delta_1 < \delta_2 < \ldots < \dmax$, and $\dmax \leq \frac{4}{n}$, and $|\delta_{j} - \delta_{j'}| \geq \frac{1}{4n^2}$
    \item $p_j \in \left[\frac{1}{n^{4j + 3}}, \frac{1}{n^{4j + 1}}\right]$
    \item $\alpha \in [n^{-1000n - 2}, n^{-1000n}]$
\end{itemize}

\subsubsection{Reproving the propositions}
We reprove the propositions from \autoref{sec:propositions}. The propositions are sufficient to prove \autoref{thm:fullconstruction} for the full construction, giving us the analogous result for the robust construction: Greedy PI again requires at least $2^n$ iterations to arrive at the optimal policy.

Since each $\ep$ is at most $2^{-100n + 2}$, we can let $\emax = 2^{-100n + 2}$. Recall that $\dmax = \delta_{f(n)}$. Since $\delta_{f(n)} = \delta_{6n + 3} \leq \frac{2(6n+3) + 1}{4n^2} < \frac{4}{n}$, we can let $\dmax = \frac{4}{n}$.

\begin{customprop}{\ref{prop:rewcost}} \label{customprop:rewcost}
$r(i) - c(i) - 2n\emax - \dmax> \sum_{j < i} r(j)$ for all $i$.
\end{customprop}

\begin{proof}
Since $r(i) \geq 2^{7(i-1) + 4}$ and $c(i) \leq 2^{7(i-1)+3}$, $r(i) - c(i) \geq 2^{7(i-1)+3}$.
Since $r(j) \leq 2^{7(j-1) + 6}$, $\sum_{j < i} r(j) \leq \sum_{j \leq i-1} 2^{7(j-1) + 6} \leq 2^{7(i-2) + 7} - 2 = 2^{7(i-1)} - 2$.
Since $\emax + \dmax < 2$, $r(i) - c(i) - \emax - \dmax \geq 2^{7(i-1)+3} > 2^{7(i-1)} - 2 \geq \sum_{j < i} r(j)$. 
\end{proof}

\begin{customprop}{\ref{prop:cost}} \label{customprop:cost} 
$c(i) > \delta_{\text{max}} + 2n\emax + \sum_{j < i} r(j)$ for all $i$.
\end{customprop}

\begin{proof}
From the proof of the previous proposition, we have $\sum_{j < i} r(j) \leq 2^{7(i-1)} - 2$, and $2n\emax + \dmax < 2$. Thus, $\dmax + 2n\emax + \sum_{j < i} r(j) \leq 2^{7(i-1)}$. 
Since $c(i) \geq 2^{7(i-1) + 1}$, this is strictly less than $c(i)$. 
\end{proof}

\begin{customprop}{\ref{prop:maxval}} \label{customprop:maxval}
When the weak invariant is satisfied, $\val(b_1) \leq r(1) + \sum_{\substack{i \in B \\ i > 1}} r(i) - c(i)$. 
\end{customprop}

\begin{proof}
The original proof of this proposition for the full construction is not dependent on the exact values of the rewards or costs and thus still holds.
\end{proof}

\begin{customlemma}{\ref{lemma:wb}} \label{customlemma:wb}
When the weak invariant is satisfied, $\val(b_1) = \val(w_1)$.
\end{customlemma}

\begin{proof}
The original proof of this proposition for the full construction is not dependent on the exact values of the rewards or costs and thus still holds.
\end{proof}

\begin{customlemma}{\ref{lemma:action1}} \label{customlemma:action1}
Let the weak invariant be satisfied, where $b_i = a^i_{f(i)}$. Then for any parameters values for which Propositions 5-8 hold, the action of $b_i$ with greatest appeal is 1.
\end{customlemma}

\begin{proof}
The original proof of this proposition for the full construction is not dependent on the exact values of the rewards or costs and thus still holds.
\end{proof}

\begin{customprop}{\ref{prop:valbound}} \label{customprop:valbound}
When the weak invariant is satisfied, for every $i \in B$ we have 
$$r(i) + \left(\sum_{\substack{j > i \\ j \in B}} r(j) - c(j)\right) - 2(n-i+1)\emax \leq \val(b_i) \leq r(i) + \left(\sum_{\substack{j > i \\ j \in B}} r(j) - c(j)\right)$$
\end{customprop}

\begin{proof}
The original proof of this proposition for the full construction is not dependent on the exact values of the rewards or costs and thus still holds.
\end{proof}

 We now prove a useful lemma about the appeal of each $a^i_{j'}$ in terms of general $\delta_{j'}$.

\begin{lemma} \label{lemma:delta}
Let the weak invariant be satisfied.
Suppose the current action at $b_i$ is $a^i_j$, and let $\delta_0 = 0$. Then for any $j' \neq j$, we have $\appeal(b_i, a^i_{j'}) = \val(b_i) + p_{j'}(\delta_{j'} - \delta_j)$. 
\end{lemma}
\begin{proof}
Since the weak invariant is satisfied, if $j = 0$ without loss of generality, $\val(b_i) = \val(b_1) = \val(b_1) + \delta_0$.
Thus, $\appeal(b_i, a^i_{j'}) = (1-p_{j'})\val(b_i) + p_{j'}(\val(b_1) + \delta_{j'}) = \val(b_i) + p_{j'}(\delta_{j'} - \delta_0)$.
Otherwise, if $j \neq 0$, we have:
\begin{align*}
    \appeal(b_i, a^i_{j'}) &= (1 - p_{j'})\val(b_i) + p_{j'}(\val(b_1) + \delta_{j'})\\
    &= \val(b_i) - p_{j'}(\val(b_1) + \delta_j) + p_{j'}(\val(b_1) + \delta_{j'})\\
    &= \val(b_i) + p_{j'}(\delta_{j'}-\delta_j)
\end{align*}
\end{proof}

\begin{customprop}{\ref{prop:action1}} \label{customprop:action1}
If the weak invariant is satisfied, and $b_i = a^i_j$, $b_i$ is switchable and the action with greatest appeal is $a^i_{j+1}$ if $j \neq f(i)$, and action 1 if $j = f(i)$.
\end{customprop}

\begin{proof}
We first show that action 1 has low appeal. As in the full construction, the appeal of action 1 is 
$\appeal(b_i, 1) = (1-\alpha)\val(b_i) + \alpha(r(i) + \val(d_{i+1}))$, which is at most $\val(b_i) + \alpha(r(i) + \val(d_{i+1}))$.
By the same argument as in the proof of \autoref{prop:maxval}, $d_{i+1}$ collects at most the rewards of higher bits in $B$. 
Since $r(j') \leq 2^{7(j'-1) + 6}$ for every $j'$,
$$r(i) + \val(d_{i+1}) \leq r(i) + \sum_{j'=i+1}^n r(j') \leq 2^{7(i-1) + 6} + \sum_{j'=i+1}^n 2^{7(j'-1)+6} < 2^{7(i-1) + 6} + 2^{7n} < n^{8n}$$
Since $\alpha \leq n^{-1000n}$, $\appeal(b_i, 1) \leq \val(b_i) + \alpha(r(i) + \val(d_{i+1})) \leq \val(b_i) + n^{-1000n} \cdot n^{8n} \leq n^{-992n}$.

We now show that for any $j_1 > j$, $\appeal(b_i, a^i_{j_1}) > \appeal(b_i, 1)$. In other words, in the case that $j \neq f(i)$, the action with greatest appeal is $a^i_{j+1}$. 
By \autoref{lemma:delta}, $\appeal(b_i, a^i_{j_1}) = \val(b_i) + p_{j_1}(\delta_{j_1} - \delta_j)$. Since $\delta_{j_1} - \delta_j \geq \frac{1}{4n^2}$, and $p_{j_1} \geq \frac{1}{n^{4j_1 + 3}} \geq \frac{1}{n^{4n+3}}$, we have
\begin{equation}
\appeal(b_i, a^i_{j_1}) \geq \val(b_i) + p_{j_1}(\delta_{j_1} - \delta_j) \geq \val(b_i) + \frac{1}{n^{4j_1 + 3}} \cdot \frac{4}{n^2} \geq \val(b_i) + \frac{4}{n^{4j_1 + 5}} \label{eq:appealaj1}
\end{equation}

Since $j_1 \leq f(n) \leq 3 + 6n \leq 7n$, we can generously lower bound bound $\frac{4}{n^{4j_1 + 5}} \geq \frac{1}{n^{5j_1}} \geq n^{-35n}$.
Thus, $\appeal(b_i, a^i_{j_1}) \geq n^{-35n} > n^{-992n} \geq \appeal(b_i, 1)$.
Furthermore, since $n^{-35n} > 0$, the appeal of $a^i_{j_1}$ is greater than the current value of $b_i$, and $b_i$ is switchable to $a^i_{j_1}$.

We finally show that for any $j_2 > j_1$, $\appeal(b_i, a^i_{j_2}) < \appeal(b_i, a^i_{j_1})$. If $j_1 = 0$, this is immediately true since by \autoref{lemma:wb}, $\appeal(b_i, a^i_0) = \val(b_1) < \appeal(b_i, a^i_{j_2})$.
Otherwise, if $j_1 \neq 0$, by \autoref{lemma:delta} we have that $\appeal(b_i, a^i_{j_2}) = \val(b_i) + p_{j_2}(\delta_{j_2} - \delta_j)$. 
Since $\dmax \leq \frac{4}{n}$, $\delta_{j_2} - \delta_j \leq \frac{4}{n}$.
$p_{j_2} \leq \frac{1}{n^{4j_2+1}}$, so
$$\appeal(b_i, a^i_{j_2}) \leq \val(b_i) + \frac{1}{n^{4j_2+1}} \cdot \frac{4}{n} \leq \val(b_i) + \frac{4}{n^{4j_2 + 2}}$$
Since $j_2 \geq j_1 + 1$, $\frac{4}{n^{4j_2 + 2}} \leq \frac{4}{n^{4(j_1 + 1) + 2}} = \frac{4}{n^{4j_1 + 6}} < \frac{4}{n^{4j_1 + 5}}$.
Recall from Equation \ref{eq:appealaj1} that $\appeal(b_i, a^i_{j_1}) \geq \val(b_i) +  \frac{4}{n^{4j_1+5}}$. Putting this together, we have, $\appeal(b_i, a^i_{j_2}) < \val(b_i) + \frac{4}{n^{4j_1+5}} \leq \appeal(b_i, a^i_{j_1})$.

Thus, $b_i$ is switchable to any action $a^i_{j_1}$ for $j_1 > j$. 
These actions are decreasing in appeal, and the appeal of action 1 is less than the appeal of any of these actions. 
Actions $a^i_{j'}$ for $j' < j$ have appeal less than $\val(b_i)$, since their rewards are $\delta_{j'} < \delta_j$. This includes action $a^i_0$, which has appeal equal to $\val(b_1)$ by \autoref{lemma:wb}.
In other words, the action with greatest appeal is $a^i_{j+1}$ as desired. This completes the case where $j \neq f(i)$.

Finally, by \autoref{lemma:action1}, when $j = f(i)$ the action of $b_i$ with greatest appeal is 1.
\end{proof}

\subsubsection{Main theorem}
Because the propositions imply the main theorem from the greedy construction section, \autoref{thm:fullconstruction}, we have the analogous result for the robust construction given the same all-zero policy as in the full construction, which we state as the following lemma:

\begin{lemma} \label{thm:robustconstruction}
Let the robust construction have parameters lying in $[-2, 2]$ as specified in \autoref{sec:robustparams}, with perturbations of up to $\frac{1}{8n^2}$. When started at the all-zero policy, Greedy PI takes at least $2^n$ iterations to arrive at the optimal policy under the total reward criterion.
\end{lemma}

We now show that $2^n$ is subexponential in $N$, the total number of vertices used for the robust construction. Recall that $n$ is the number of bit vertices $b_i$; let $N$ denote the total number of vertices. 
For each bit $b_i$, we can count the number of vertices used in the various gadgets. 
\begin{itemize}
    \item $r(i)$ requires at most $10n$ vertices.
    \item $c(i)$ requires at most $10n$ vertices.
    \item $\ep$ requires at most $100n$ vertices.
    \item $a^i_j$ requires at most $4j + 3 \leq 4(n+3) + 3 \leq 5n$ vertices to create $p_j$.
    \item $\alpha$ requires $1000n + 2$ vertices.
    \item $\delta_j$ requires no additional vertices.
\end{itemize}
In the structure for $b_i$, we have one reward $r(i)$, one cost $c(i)$, two small costs $\ep$, at most $n+3$ actions $a^i_j$, and one instance of small probability $\alpha$.
Thus, the number of vertices used in the gadgets to create all of these is at most
$$7n + 7n + 2 \cdot 100n + (n+3)(5n) + 1000n + 2 = 5n^2 + 1229n + 2$$
Additionally, at each bit $b_i$, we have the ``real" vertices (not within gadgets) $w_i, b_i, c_i, d_i$, and the two square vertices $b_i^-$ and $b_i^+$; this makes 6.
Thus, at each bit $b_i$, we have at most $5n^2 + 1229n + 8$ vertices.
We have at most $n+1$ bits (we have only a partial structure for the last bit $n+1$).
This yields $N \leq (n+1)(5n^2 + 1229n + 8) \leq 6n^3$ for $n$ sufficiently large.

We previously showed that Greedy PI takes at least $2^n$ iterations. In terms of $N$, this is $2^{\sqrt[3]{\frac{N}{6}}}$ iterations.
Recall that our perturbations were up to $\frac{1}{8n^2}$. Since $n \leq \sqrt[3]{\frac{N}{6}}$, perturbations of $\frac{1}{N}$ are at most $\frac{1}{8n^2}$ for sufficiently large $n$ and  $N$.

\begin{theorem}
Let the robust construction have parameter values as specified in \autoref{sec:robustparams}, with allowed perturbations of up to $\frac{1}{N}$, where $N$ is the number of vertices. Starting at the all-zero policy, Greedy PI takes at least $2^{\sqrt[3]{\frac{N}{6}}}$ iterations to arrive at the optimal policy under the total reward criterion.
\end{theorem}

By the same argument as in \autoref{subsec:full}, since the MDP always terminates at a 0-reward sink state, we have as a corollary the same result for the average reward criterion:
\begin{corollary}
Let the robust construction have parameter values as specified in \autoref{sec:robustparams}, with allowed perturbations of up to $\frac{1}{N}$, where $N$ is the number of vertices. Starting at the all-zero policy, Greedy PI takes at least $2^{\sqrt[3]{\frac{N}{6}}}$ iterations to arrive at the optimal policy under the average reward criterion.
\end{corollary}

\section{A lower bound for Greedy PI under the reachability criterion}
In this section, we prove an exponential lower bound for the reachability criterion in the worst case, without perturbations. We start in \autoref{subsec:newparam} by presenting new parameter ranges for the full construction, then reproving Propositions 5-9 to show that Greedy PI takes exponentially many iterations. 
We give ranges here instead of explicit values to simplify our proof later; we do not consider perturbations here.
Then, in \autoref{subsec:reach}, we show how to modify the full construction with gadgets such that under the reachability criterion, the effective rewards lie in those ranges.
It then follows that Greedy PI requires exponentially many iterations on this modified construction under the reachability criterion.

\subsection{New parameter ranges for the full construction} \label{subsec:newparam}

Suppose that the parameters from the full construction under the total reward criterion be in the following ranges: 
\begin{itemize}
    \item $r(i) \in [2^{4i + 2 - 5n}, 2^{4i + 3 - 5n}]$
    \item $c(i) \in [2^{4i - 5n}, 2^{4i + 1 - 5n}]$
    \item $\ep \in [2^{-100n -1}, 2^{-100n}]$
    \item $\delta_j \in [2^{-200n + 2j} , 2^{-200n + 2j + 1} ]$.
    $0 < \delta_1 < \delta_2 < \ldots < \dmax$, and $\dmax \leq 2^{-100n}$, and $|\delta_{j} - \delta_{j'}| \geq 2^{-200n} $
    \item $\alpha = 2^{-400n -400f(n)^2}$
    \item $p_j = 2^{-400j^2}$
\end{itemize}

Before reproving the propositions, we first prove two useful inequalities.

Recall that $r(j) \leq 2^{4j+3-5n}$ for any $j$. Thus, for any $i$, 
\begin{equation}
    \sum_{j < i} r(j) \leq 2^{-5n}\sum_{j < i} 2^{4j+2} < 2^{-5n} \cdot 2^{4(i-1) + 2 + 1} = 2^{-5n} \cdot 2^{4i-1} \label{eq:reachsumrew}
\end{equation}

Recall that $\dmax \leq 2^{-100n}$ and $\emax \leq 2^{-100n}$. Thus,
\begin{equation}
    \dmax + 2n\emax \leq (2n + 1)2^{-100n} < 2^n \cdot 2^{-100n} = 2^{-99n} \label{eq:reachsmallrews}
\end{equation}

We now reprove the propositions from \autoref{sec:propositions}, which together imply that Greedy PI takes $2^n$ iterations. Recall that propositions \ref{prop:maxval} and \ref{prop:valbound} and lemmas \ref{lemma:wb} and \ref{lemma:action1} don't depend on the exact values of the rewards or costs and thus don't need to be reproven.

\begin{customprop}{\ref{prop:rewcost}} \label{customprop2:rewcost}
$r(i) - c(i) - 2n\emax - \dmax > \sum_{j < i} r(j)$ for all $i$.
\end{customprop}

\begin{proof}
By Equation \ref{eq:reachsmallrews}, $\dmax + 2n\emax < 2^{-99n}$. 
For any $i$, $r(i) - c(i) \geq 2^{4i+2-5n} - 2^{4i+1-5n} = 2^{4i+1-5n}$.
Thus, $r(i) - c(i) - 2n\emax - \dmax > 2^{4i+1-5n} - 2^{-99n} > 2^{4i - 5n}$.

By Equation \ref{eq:reachsumrew}, $\sum_{j < i} r(j) < 2^{4i-1-5n}$. Thus, $r(i) - c(i) - 2n\emax - \dmax > \sum_{j < i} r(j)$.
\end{proof}

\begin{customprop}{\ref{prop:cost}} \label{customprop2:cost} 
$c(i) > \delta_{\text{max}} + 2n\emax + \sum_{j < i} r(j)$ for all $i$.
\end{customprop}

\begin{proof}
By Equation \ref{eq:reachsumrew}, $\sum_{j < i} r(j) < 2^{-5n} \cdot 2^{4i-1}$.
By Equation \ref{eq:reachsmallrews}, $\dmax + 2n\emax < 2^{-99n}$.
The sum of these two expressions is then less than 
$\dmax + 2n\emax + \sum_{j < i} r(j) < 2^{4i-1-5n} + 2^{-99n} < 2^{4i-5n} \leq c(i)$.
\end{proof}

\begin{customprop}{\ref{prop:action1}} \label{customprop2:action1}
If the weak invariant is satisfied, and $b_i = a^i_j$ where $j \neq f(i)$, $b_i$ is switchable and the action with greatest appeal is $a^i_{j+1}$, and action 1 if $j = f(i)$.
\end{customprop}

\begin{proof}
We first consider the case where $j \neq f(i)$.
We begin by showing that action 1 has low appeal. $\appeal(b_i, 1) = (1-\alpha)\val(b_i) + \alpha(r(i) + \val(d_{i+1}))$.
By the same argument as in the proof of \autoref{prop:maxval}, only the rewards $r(j)$ for $j > i$ are collectible from $d_{i+1}$, so $\val(d_{i+1}) \leq \sum_{j > i} r(j) \leq \sum_{j=1}^n 2^{4i+2-5n} < 2^{4n+3-5n}$.
Since $r(i) \leq 2^{4i+3-5n}$, $r(i) + \val(d_{i+1}) < 1$.
Since $\alpha = 2^{-400n -400f(n)^2}$, $\appeal(b_i, 1) \leq \val(b_i) + \alpha(r(i) + \val(d_{i+1})) < \val(b_i) + 2^{-400n -400f(n)^2}$.

Let $j_2 > j_1 > j$.
Note first that by \autoref{lemma:wb}, $\appeal(b_i, a^i_0) = \val(b_1)$, so action $a^i_0$ has an effective reward of $\delta_0 = 0$.
By \autoref{lemma:delta}, for any $j' \neq j$, we have $\appeal(b_i, a^i_{j'}) = \val(b_i) + p_{j'}(\delta_{j'} - \delta_j)$. 
Since $\delta_{j_2} \leq 2^{-200n + 2j_2 + 1}$, $\appeal(b_i, a^i_{j_2}) \leq \val(b_i) + p_{j_2} \cdot \delta_{j_2} \leq \val(b_i) + 2^{-200n - 400j_2^2 + 2j_2 + 1}$.
Since $j_2^2 \geq (j_1 + 1)^2\geq j_1^2 + 2j_1$, 
$2^{-200n - 400j_2^2 + 2j_2 + 1} \leq 2^{-200n - 400j_1^2 - 800j_1 + 2j_1 + 1} < 2^{-200n - 400j_1^2}$.
The difference between $\delta_{j_1}$ and $\delta_j$ is at least $2^{-200n}$, so 
$\appeal(b_i, a^i_{j_1}) \geq \val(b_i) + p_{j_1} \cdot 2^{-200n}$.
Since $p_{j_1} = 2^{-400j_1^2}$, we have $p_{j_1} \cdot 2^{-200n} \geq 2^{-200n - 400j_1^2}$.
Thus, $\appeal(b_i, a^i_{j_1}) > \appeal(b_i, a^i_{j_2})$.

Note also that $\appeal(b_i, a^i_{j_1}) \geq \val(b_i) +  2^{-200n - 400j_1^2} \geq \val(b_i) + 2^{-200n - f(n)^2}$, since $j_1 \leq f(n)$.
This is strictly less than the appeal of action 1, which is at most $\val(b_i) + 2^{-400n - 400f(n)^2}$.
Observe also that $\appeal(b_i, a^i_{j_1}) > \appeal(b_i, a^i_j)$ only if $\delta_{j_1} > \delta_j$. This is true even for $j_1 = 0$, since $\appeal(b_i, a^i_0) = \val(b_1)$ by \autoref{lemma:wb}.
Thus, the switchable actions are exactly the action 1 and the actions $a^i_{j'}$ for $j' > j$.
We have shown that these actions have decreasing appeal as $j'$ increases.
Thus, when $j \neq f(i)$, the action with greatest appeal is $a^i_{j+1}$.

Finally, by \autoref{lemma:action1}, when $j = f(i)$, the action with greatest appeal is action 1.
\end{proof}

\subsection{Transforming to the reachability criterion} \label{subsec:reach}
Policy iteration under the reachability criterion maximizes the probability of reaching a designated sink node. The value of a node under a policy is the probability of reaching the sink from that node. 
Thus, policy iteration under the reachability criterion ignores rewards and costs.
To adapt our full construction for the reachability criterion, we introduce gadgets that simulate rewards and costs using random actions.
These gadgets require known bounds on the minimum and maximum values of any nodes.

\begin{claim}
Given the parameters in this section, the maximum value of any node is at most $\frac{1}{4}$.
\end{claim}

\begin{proof}
Rewards can be collected only at the nodes $b_i$. No policy reached by Greedy PI has actions forming a cycle, because of the costs $c(j)$. Thus, each $b_i$ is visited at most once, and the maximum value of any node is at most the sum of rewards collected at each $b_i$.
The greatest reward of any action at $b_i$ is $r(i)$.
Thus, the maximum value is at most 
$$\sum_{i=1}^n r(i) \leq \sum_{i=1}^n 2^{4i + 3 - 5n} \leq 2^{4n + 4 - 5n} = 2^{4-n} < \frac{1}{4}$$
\end{proof}

Now, observe that if we add a reward of $\frac{1}{4}$ upon reaching the sink, the behavior of policy iteration is not affected. 
This is because it increases the value of every vertex by exactly $
\frac{1}{4}$, and also increases the appeal of every action by $
\frac{1}{4}$. 
Thus, adding this reward does not affect any appeals or values relative to each other.
Since policy iteration depends only on relative appeals, its behavior does not change. 

After adding this reward of $\frac{1}{4}$ upon reaching the sink, the value of every node lies in $[\frac{1}{4}, \frac{1}{2}]$ throughout the duration of policy iteration. We use these bounds in our transformation to the reachability criterion, using the following gadgets.

\begin{figure}
    \centering
    \includegraphics[width=400pt]{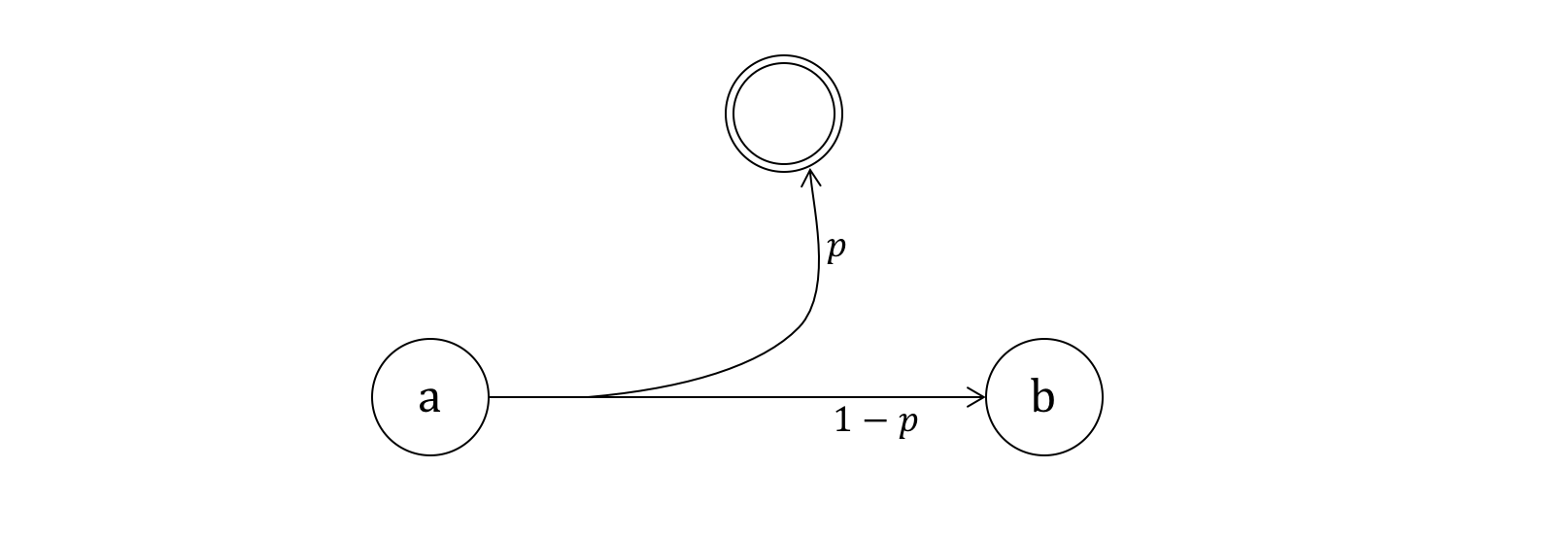}
    \caption{Gadget $g_4(p)$ for simulating positive rewards under the reachability criterion. The vertex with a double-line border is the sink that we are trying to reach. $p$ is the probability of reaching the sink from the depicted action.}
    \label{fig:reachability}
\end{figure}

\begin{lemma}
Assume that the value of every vertex lies in $\left[\frac{1}{4}, \frac{1}{2}\right]$. 
Then, for any $p \in [0, 1]$ and any deterministic action between vertices $a$ and $b$, $g_4(p)$, shown  in \autoref{fig:reachability}, achieves $\val(a) = \val(b) + r$ where $r \in \left[\frac{p}{2}, \frac{3p}{4}\right]$.
\end{lemma}

\begin{proof}
We have $\val(a) = (1-p)\val(b) + p$. Since $\val(b) \leq \frac{1}{2}$, we have $\val(a) = \val(b) + p - p\val(b) \geq \val(b) + \frac{p}{2}$.
Since $\val(b) \geq \frac{1}{4}$, we have $\val(a) = \val(b) + p - p\val(b) \leq \val(b) + \frac{3p}{4}$.
\end{proof}

\begin{figure}
    \centering
    \includegraphics[width=400pt]{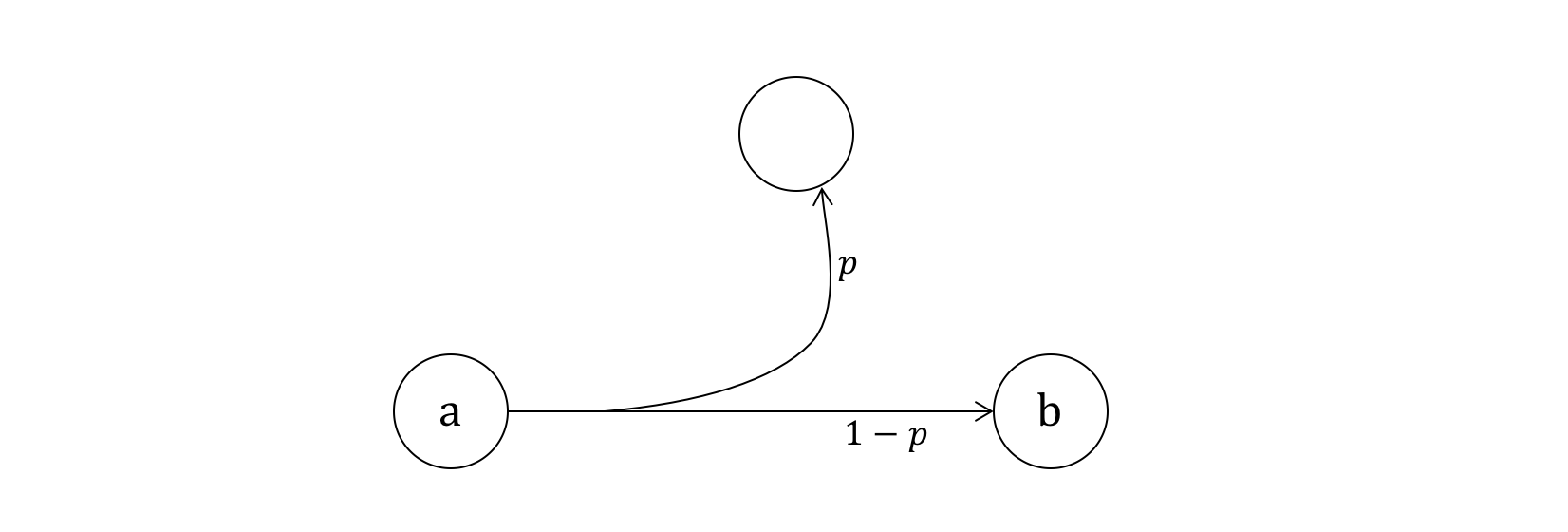}
    \caption{Gadget $g_5(p)$ for negative rewards under the reachability criterion. The top vertex is the 0-sink that we get no reward for reaching. $p$ is the probability of reaching the sink from the depicted action.}
    \label{fig:g5}
\end{figure}

We also create a gadget $g_5(p)$ that simulates a negative reward, shown in \autoref{fig:g5}.

\begin{lemma}
Assume that the value of every vertex lies in $\left[\frac{1}{4}, \frac{1}{2}\right]$. 
Then, for any $p \in [0, 1]$ and any deterministic action between vertices $a$ and $b$, $g_5(p)$ achieves $\val(a) = \val(b) - r$ where $r \in \left[\frac{p}{4}, \frac{p}{2}\right]$.
\end{lemma}

\begin{proof}
We have $\val(a) = (1-p)\val(b) + p \cdot 0$. Since $\val(b) \leq \frac{1}{2}$, $(1-p)\val(b) \geq \val(b) - \frac{p}{2}$.
Since $\val(b) \geq \frac{1}{4}$, $(1-p)\val(b) \leq \val(b) - \frac{p}{4}$.
\end{proof}

We can convert the full construction to use the reachability criterion by using gadget $g_4$ in place of any action with a positive reward, and using gadget $g_5$ in place of any action with a negative reward. 
We slightly edit the actions $a^i_j$ so that the rewards fall on deterministic actions, in order to apply the gadgets. This does not affect the behavior of policy iteration.

The full construction has $N = O(n)$ nodes in total, where $n$ is the number of bit-nodes $b_i$.
This conversion to the reachability criterion introduces at most $O(N)$ new nodes. Thus, we still have $O(N)$ vertices in total, giving us the following lower bound:

\begin{theorem} \label{thm:fullconstruction2}
There exists an MDP on $N$ nodes on which Greedy PI under the reachability criterion takes $2^{\Omega(N)}$ iterations to arrive at the optimal policy.
\end{theorem}

\section{Smoothed lower bounds for simple, difference, and topological policy iteration}
Melekopoglou and Condon \cite{melekopoglou1994complexity} constructed MDPs on which policy iteration requires exponential time in the worst case when using the simple, topological, and difference switching rules. The MDPs in these constructions are reachability MDPs.
The construction involves two sink nodes, labeled $0^*$ and $1^*$. The goal is to minimize the probability of reaching $1^*$. Equivalently, every edge to $1^*$ incurs a cost of 1, and all other edges have a cost of 0. 
There are $2n + 1$ other vertices: $0', 1', \ldots, n'$, and $1, 2, \ldots, n$. The vertex $0'$ has a random edge going to $n$ with probability $\frac{1}{2}$ and going to $1^*$ with probability $\frac{1}{2}$. The vertex $1'$ has a random edge going to $1^*$ with probability $\frac{1}{2}$ and to $0^*$ with probability $\frac{1}{2}$. Each other vertex $k'$ has a random edge going to $k-2$ with probability $\frac{1}{2}$ and to $(k-1)'$ with probability $\frac{1}{2}$. 
Each vertex $k$ has two deterministic actions: an action 0 going to $k-1$, and an action 1 going to $k'$.
We call the vertices $1, \ldots, n$ \textit{min-vertices}, and we call the vertices $1', \ldots, n'$ \textit{random vertices}.

Let $S_k$ denote the action chosen at vertex $k$. We write a policy as a string $S_n S_{n-1} \ldots S_1$. Thus, the policy where every vertex $k$ takes action 0 is denoted $00 \ldots 0$.
Observe that the optimal policy is $00 \ldots 01$, with vertex 1 set to action 1, and all other vertices set to action 0.

\begin{figure} 
    \centering
    \includegraphics[width=400pt]{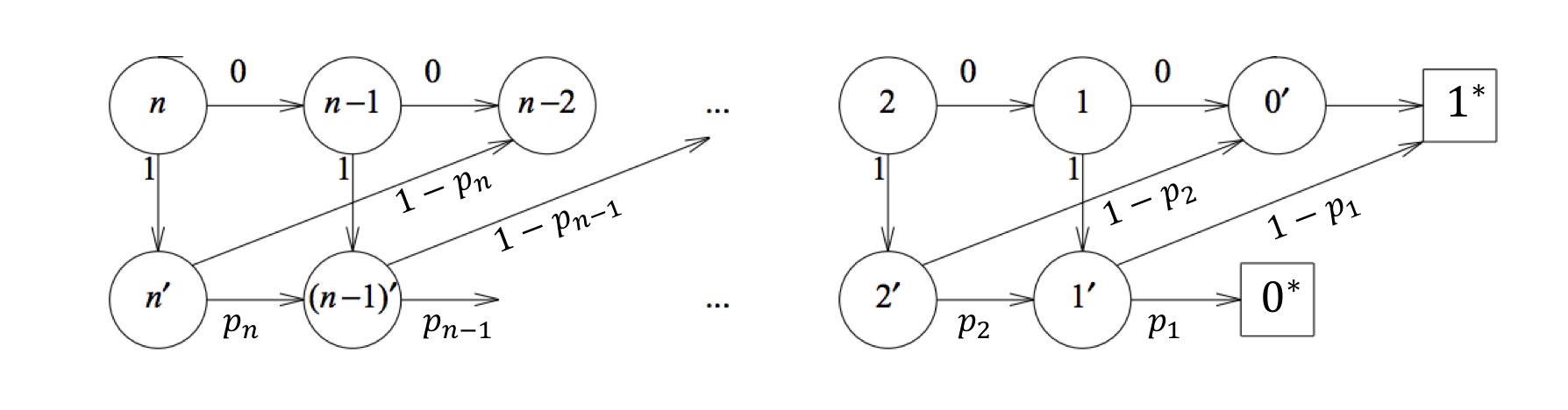}
    \caption{The basic graph on $n$ vertices}
    \label{fig:basicgraph}
\end{figure}

We use slightly different notation here than in the earlier sections, because the MDPs we consider are 2-action MDPs.
We use $V(k)$ or $V(k')$ to denote the cost of a vertex $k$ or $k'$, i.e. the probability of reaching the target sink $1^*$, starting from vertex $k$ or $k'$; this is the same notation used in the original paper \cite{melekopoglou1994complexity}.
For each min-vertex $k$, we use $\diff(k)$ to denote the difference in value between the children of $k$. This is well-defined for this construction since each vertex $k$ has at most two children. More precisely, we let $\diff(1) := V(1') - V(0')$, and for $k \geq 2$,
$$\diff(k) := V(k') - V(k-1)$$

We show that with slight modifications, the Melekopoglou-Condon (MC) construction requires exponential time even when the probabilities assigned to random edges can be perturbed.

\subsection{Simple Policy Improvement Algorithm} \label{subsec:simple}
Simple policy iteration involves switching the action of exactly one vertex in each iteration. If there are multiple switchable vertices, the vertex whose action is switched is chosen according to a fixed ordering. 
Here, we choose the highest-labeled switchable vertex to be switched.

We use the MC basic construction but remove the edge from $0'$ to $n$, and make the edge from $0'$ to $1^*$ have probability 1. We call this modified version of the MDP the \emph{basic graph}, shown in \autoref{fig:basicgraph}.
We introduce variables $p_k$ to represent the probabilities of the random edges.
Let $p_1$ denote the probability of the edge from $1'$ to $0^*$, meaning the edge from $1'$ to $1^*$ has probability $1-p_1$.
For $k \geq 2$, let the edge from vertex $k'$ to $(k-1)'$ have probability $p_k$ and let the edge from $k'$ to $k-2$ have probability $1-p_k$. 
Let $S_k$ denote the state of vertex $k$.

Our proof follows the same format as the analogous one of MC. First, we derive an expression for $\diff(k)$. Then we use this expression to reprove lemmas from MC, which are sufficient to prove the lower bound.

\begin{lemma} \label{lemma:diff}
For every $k \geq 1$, $\diff(k) = \diff(1)\prod_{i=2}^k (p_i - S_{i-1})$.
\end{lemma}
\begin{proof}
We derive several recursive expressions for the costs of the vertices, which will give us a nice formula for $\diff(k)$.
From the structure of the graph, we have
\begin{align*}
    V(k') &= p_kV((k-1)') + (1-p_k)V(k-2)\\
    V(k-1) &= (1-S_{k-1})V(k-2) + S_{k-1}V((k-1)')
\end{align*}
Plugging in $V(k')$ and $V(k-1)$ in the expression for $\diff(k)$, we have
\begin{align}
    \diff(k) &= V(k') - V(k-1) \nonumber \\
    &= (p_k - S_{k-1})V((k-1)') + (1 - p_k - 1 + S_{k-1})V(k-2) \nonumber \\
    &= (p_k - S_{k-1})(V((k-1)') - V(k-2)) \nonumber \\
    &= (p_k - S_{k-1})\diff(k-1) \label{eq:diff}
\end{align}
We can make the following inductive argument. For the base case, $\diff(1) = \diff(1)$. Assuming that $\diff(k-1) = \diff(1)\prod_{i=2}^{k-1} (p_i - S_{i-1})$, we have by \autoref{eq:diff} that $\diff(k) = (p_k - S_{k-1})\diff(k-1) = \diff(1)\prod_{i=2}^k (p_i - S_{i-1})$ as desired.
\end{proof}

Observe that $k$ is switchable if and only if $S_k = 0$ and $\diff(k) < 0$ or $S_k = 1$ and $\diff(k) > 0$, since the goal is to minimize the cost. We can use \autoref{lemma:diff} to compute $\text{sgn}(\diff(k))$ and thus determine whether any vertex $k$ is switchable given the actions of the other states.
The original proof from \cite{melekopoglou1994complexity} that policy iteration requires exponentially many iterations still holds because Lemmas 2.6 and 2.7 from \cite{melekopoglou1994complexity} still hold. 
We state these below as Lemmas \ref{lemma:switch1} and \ref{lemma:switch2}. 

\begin{lemma} \label{lemma:switch1}
If $k$ is switched, and $S_n \ldots S_{k+2}S_{k+1} = 0 \ldots 01$, then all the vertices $k+1, k+2, \ldots, n$ are switchable.
\end{lemma}

\begin{proof}
If $S_k$ is switched from 0 to 1, then we must have had $\diff(k) < 0$ before switching. Since $p_i - S_{i-1} > 0$ for all $i > k+2$,  $p_{k+2} - S_{k+1} < 0$, and $p_{k+1} - S_{k} < 0$ after $k$ is switched, we have $\text{sgn}(\diff(i)) = \text{sgn}(\diff(k))$ for all $i > k+1$. Thus for all $i > k+1$, $\diff(i) < 0$, and all of these vertices are switchable. 
For $k+1$, we have $\text{sgn}(\diff(k+1)) = -\text{sgn}(\diff(k)) > 0$, and since $S_{k+1} = 1$, $S_{k+1}$ is also switchable.

If $S_k$ is switched from 1 to 0, $\diff(k) > 0$. Again, we have that $p_i - S_{i-1} > 0$ for all $i > k+1$, and $p_{k+2} - S_{k+1} < 0$. We now have that $p_{k+1} - S_k > 0$, since $S_k$ is switched from 1 to 0. 
Thus $\text{sgn}(\diff(i)) = -\text{sgn}(\diff(k))$ for all $i > k+2$, after $S_k$ is switched. Thus these all have $\diff(i) < 0$ and $S_i = 0$ and are switchable. 
For vertex $k+1$, we have $\text{sgn}(\diff(k+1)) = \text{sgn}(\diff(k)) > 0$ and $S_{k+1} = 1$, so $S_{k+1}$ is also switchable.
\end{proof}

\begin{lemma} \label{lemma:switch2}
The following two statements hold for every positive $k$:
\begin{enumerate}[(1)]
    \item If $S_n \ldots S_{k+1}S_k = 0 \ldots 01$, and the vertices $k, k+1, \ldots, n$ are switchable, the next $2^{n-k+1} - 1$ switches of the simple policy improvement algorithm are made on these vertices, to reach the policy where $S_n \ldots S_{k+1}S_k = 0 \ldots 00$.
    \item If $S_n \ldots S_{k+1}S_k = 0 \ldots 00$, and the vertices $k, k+1, \ldots, n$ are switchable, the next $2^{n-k+1}-1$ switches of the simple policy improvement algorithm are made on these vertices, to reach the policy where $S_n \ldots S_{k+1}S_k = 0 \ldots 01$.
\end{enumerate}
\end{lemma}

\begin{proof}
We prove the lemma by induction on $k$, starting from $k = n$ and decreasing. For our base case, $k = n$, statements 1 and 2 trivially hold since if $n$ is switchable, one switch is made on $n$.

Assume that for all $k \geq m$, the two statements hold. We will show that both statements hold for vertex $m-1$.
\par{Statement 1.} If $S_n \ldots S_mS_{m-1} = 0 \ldots 01$, and $m-1, m, \ldots, n$ are switchable, then $(2)$ is satisfied for $k = m$. Thus the next $2^{n-m+1}-1$ switches are made so that $S_n \ldots S_{m+1} S_m S_{m-1} = 0 \ldots 011$. Now, $m-1$ is still switchable since no lower number vertices were switched, and $\diff(m-1) > 0$. 
Since $(p_m - S_{m-1}) < 0$, $(p_i - S_{i-1}) > 0$ for $i > m + 1$, and the vertices $m, m+1, \ldots, n$ are no longer switchable, $m-1$ is now switched. 
This gives us $S_n S_{n-1} \ldots S_m S_{m-1} = 0 0 \ldots 1 0$. 
By Lemma \ref{lemma:switch1}, $m, \ldots, n$ are all switchable again, and $(1)$ from the inductive hypothesis holds for $k = m$. So the next $2^{n-m+1}-1$ switches are made to get $S_n \ldots S_m S_{m-1} = 0 \ldots 00$. 
The total number of switches made is $2^{n-m+1} - 1 + 1 + 2^{n-m+1} - 1= 2^{n-(m-1)+1}-1$ as desired.

\par{Statement 2.} The proof of statement 2 follows similarly.
\end{proof}

\begin{theorem} \label{thm:simple}
Given the basic graph on $n$ states, the simple policy improvement algorithm requires $2^n$ iterations in the worst case, even when the probabilities associated with the probabilistic vertices can be perturbed within the open interval $(0, 1)$ and the cost of the sink $1^*$ can be perturbed to any positive value.
\end{theorem}

\begin{proof}
Consider the basic graph from \autoref{fig:basicgraph}, with probabilities $p_1, \ldots, p_n$ for the probabilistic vertices. 
If we start with policy $00 \ldots 00$, all vertices are switchable because $\diff(1) = -p_1 < 0$ and $p_i - S_{i-1} > 0$ for all $i \geq 2$.
So by Lemma \ref{lemma:switch2}, the simple policy improvement algorithm makes $2^{n}-1$ switches to arrive at the optimal policy $00 \ldots 01$.

Let $c$ denote the cost of the sink $1^*$.
To show that the cost can be perturbed, recall the expression of $\diff$ from \autoref{lemma:diff}: for every $k \geq 1$, $\diff(k) = \diff(1)\prod_{i=2}^k (p_i - S_{i-1})$.
$\diff(1) = V(1') - V(0') = c(1-p_1) - c = -p_1 c$.
Thus, as long as $c$ is positive, $\diff(1)$ is negative and the argument remains unchanged, since for any vertex $k$, the sign of $\diff(k)$ depends only on the signs of $\diff(i)$ for $i < k$.

\end{proof}

\subsection{Topological Policy Improvement Algorithm}
Topological policy iteration, like simple policy iteration, switches one vertex in each iteration. If multiple vertices are switchable, the vertex to be switched is selected based on a topological ordering: if there is a path from a vertex $i$ to a vertex $j$, the order of $i$ is at least the order of $j$. Topological policy iteration always switches a vertex of the lowest order; if there are multiple switchable vertices in that order, it picks the highest-numbered vertex among those.

\begin{figure}
    \centering
    \includegraphics[width=400pt]{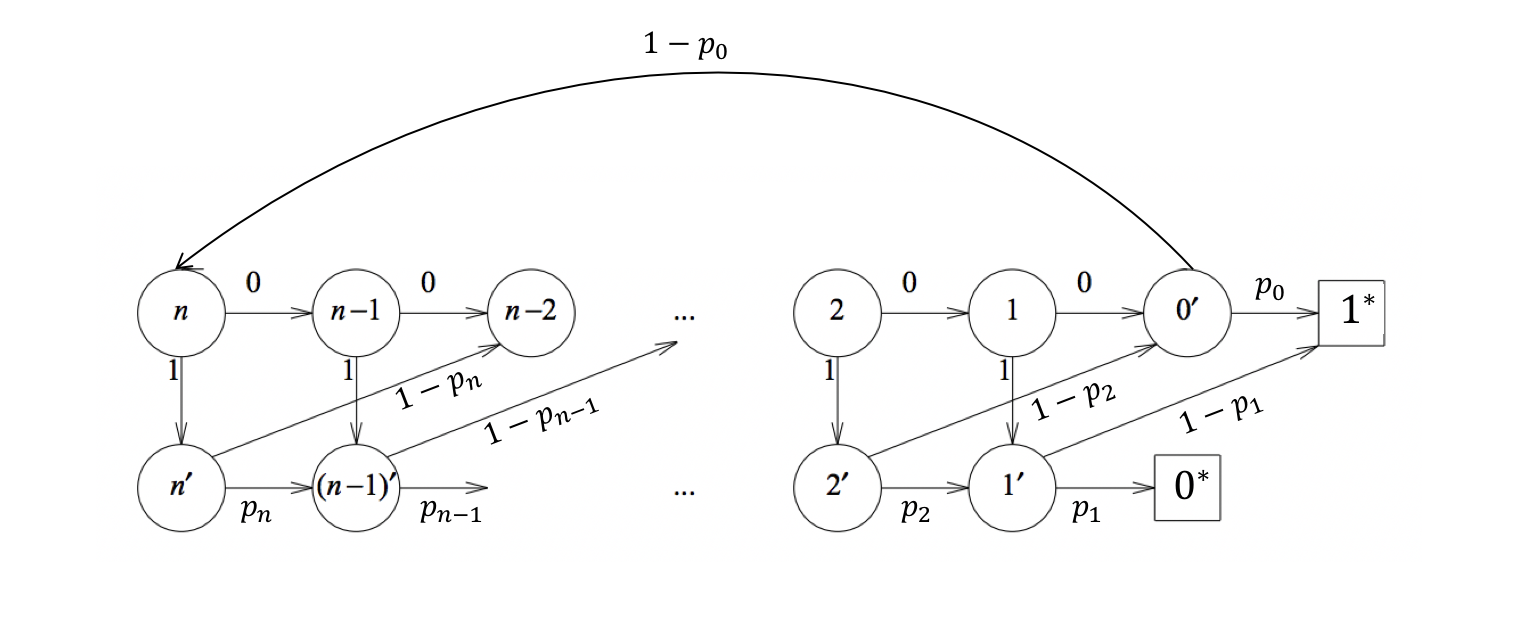}
    \caption{The topological graph on $n$ vertices}
    \label{fig:topologicalgraph}
\end{figure}

For this section, we add back in the edge from $0'$ to $n$ and let $p_0$ be the probability of the edge from $0'$ to $1^*$. We call this the \emph{topological graph}, shown in \autoref{fig:topologicalgraph}.
Now, all non-sink vertices have the same order, and thus on this MDP topological policy improvement makes the same switches as simple policy improvement.
It turns out that adding back the edge from $0'$ to $n$ does not affect the switches switches made by Simple PI (and therefore Topological PI), which we argue to obtain the following theorem.

\begin{theorem} \label{thm:topological}
Given the topological graph on $n$ states, the topological policy improvement algorithm requires $2^n$ iterations in the worst case, even when the probabilities and cost are perturbed, as long as we have that $p_0, p_1, \ldots, p_n \in (0, 1)$, the probabilities $p_0$ and $p_1$ satisfy $p_0 > 1-p_1$, and the cost $c$ of the sink $1^*$ is positive.
\end{theorem}

\begin{proof}
Recall from the previous section that in that graph, $\diff(k) = \diff(1)\prod_{i=2}^k (p_i - S_{i-1})$. This still holds when the edge from $0'$ to $n$ is added back in, since this does not affect the relationship between $\diff(i)$ and $\diff(i-1)$ for $i > 1$.
With the edge added back in and cost $c > 0$ associated with the 1-sink, we have 
$$\diff(1) = V(1') - V(0) = c(1 - p_1) - p_0 c - (1-p_0)V(n) \leq c(1 - p_1 - p_0)$$
If $p_0 > 1-p_1$, this expression is always negative, and we again have the property that vertex $k$ is switchable if and only if $S_k = 0$ and $\diff(k) < 0$ or $S_k = 1$ and $\diff(k) > 1$.
Then by the same argument as in the previous section, the algorithm with a switching policy that chooses the largest numbered vertex first uses an exponential number of iterations.
\end{proof}

\subsection{Difference Policy Improvement Algorithm} \label{subsec:difference}
The difference policy improvement algorithm, like simple and topological policy iteration, switches one vertex in each iteration. 
This vertex is chosen to maximize the difference between the costs of its two children; the chosen vertex $v^*$ from the set of switchable min-vertices $V$ satisfies
$$v^* = \argmax_{v \in V} |\diff(v)|$$
In this section, we use the basic graph, and we insert gadgets used in \cite{melekopoglou1994complexity}, with different parameters to account for the perturbation.

\begin{figure}
    \centering
    \includegraphics[width=400pt]{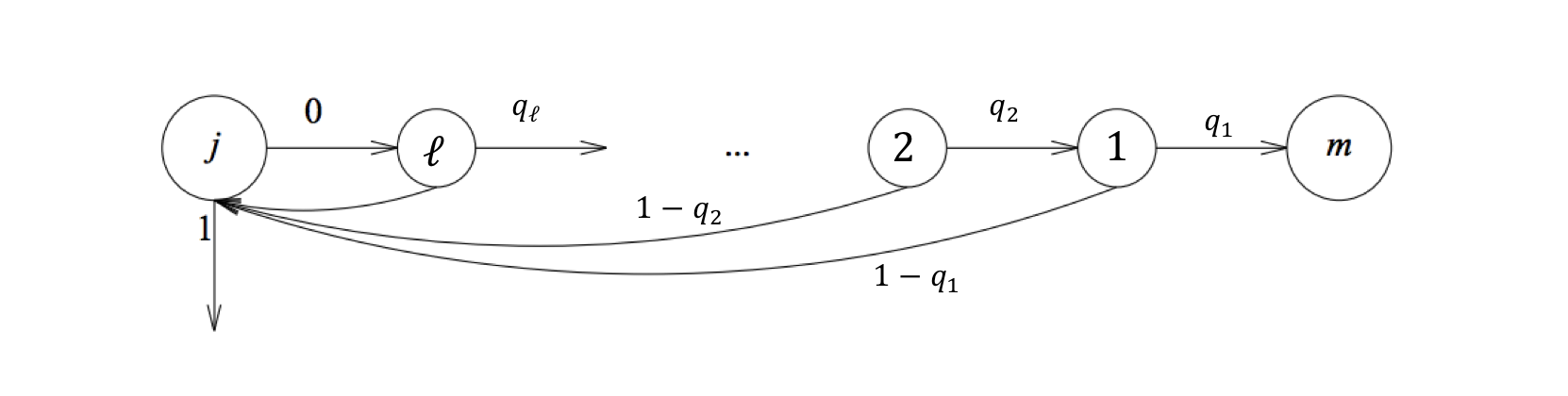}
    \caption{The gadget $g_1(\ell)$ from $j$ to $m$, with $\ell$ intermediate nodes. The edges shown from $j$ are actions 0 and 1. From each intermediate gadget node $i$, there is a single zero-reward probabilistic action going to $(i-1)$ with probability $q_i$ and back to $j$ with probability $1-q_i$.}
    \label{fig:g1}
\end{figure}

For each min-vertex $k$, we add two copies of the gadget $g_1({f(k))}$ shown in Figure 3 to the basic graph where $f$ is a function that is defined later: one copy of the gadget is inserted between $k$ and $k'$, and one between $k$ and $k-1$, replacing the corresponding edges of the basic graph.
We assume that the probabilities $q_i$ within the gadgets are between $\frac{1}{2}$ and $\frac{1}{2} + \frac{1}{n}$.

\begin{defin}
Let $k$ be a min-vertex. Let $a$ be the child of $k$ in the gadget between $k$ and $k-1$; let $b$ be the child of $k$ in the gadget between $k$ and $k'$. We define
$$\diff'(k) := V(b) - V(a)$$
\end{defin}

Let $q_1, \ldots, q_{f(k)}$ be the probabilities in the gadget from $k$ to $k-1$, and let $r_1, \ldots, r_{f(k)}$ be the probabilities in the gadget from $k$ to $k'$. Let $a$ be the child of $k$ in the gadget between $k$ and $k-1$, and let $b$ be the child of $k$ in the gadget between $k$ and $k'$.
We first derive expressions for the values of $\diff$ yielded by the gadget in \autoref{lemma:diffbound}. 
We then use this expression to show that $\diff$ is increasing in $n$, so difference policy iteration follows the same sequence of policies as simple policy iteration, flipping higher-numbered bits first. 

\begin{lemma} \label{lemma:diffbound}
If $S_k = 1$, $\diff'(k) = \left(\prod_{i=1}^{f(k)} q_i\right)\diff(k)$. If $S_k = 0$, $\diff'(k) = \left(\prod_{i=1}^{f(k)} r_i \right)\diff(k)$
\end{lemma}

\begin{proof}
If $S_k = 1$, 
$V(a) = \left(\prod_{i=1}^{f(k)} q_i\right)V(k-1) + \left(1 - \prod_{i=1}^{f(k)} q_i\right)V(k')$, and 
$V(b) = V(k')$. Then, expanding the definition of $\diff'$, we have
$$\diff'(k) = V(b) - V(a) = \left(\prod_{i=1}^{f(k)} q_i\right)(V(k') - V(k-1)) = \left(\prod_{i=1}^{f(k)} q_i\right)\diff(k)$$
If $S_k = 0$, 
$V(a) = V(k-1)$, and
$V(b) = \left(\prod_{i=1}^{f(k)} r_i\right)V(k') + \left(1 - \prod_{i=1}^{f(k)} r_i\right)V(k-1)$. Then, again expanding the definition of $\diff'$, we have
$$\diff'(k) = V(b) - V(a) = \left(\prod_{i=1}^{f(k)} r_i\right)(V(k') - V(k-1)) = \left(\prod_{i=1}^{f(k)} r_i\right)\diff(k)$$
\end{proof}

\begin{theorem} \label{thm:difference}
Let $f(n) = 0$, and let $f(k) = \lceil \log_{(\frac{1}{2} + \frac{1}{n})}{((\frac{1}{2})^{f(k+1)} \cdot (\frac{1}{3}))} \rceil$ for $1 \leq k < n$. On the MDP obtained from the basic graph by adding gadgets $g_1({f(v)})$ between each min-vertex $v$ and its two children, the difference policy improvement algorithm requires $2^n$ iterations, as long as all the probabilities $q_j$ and $r_j$ for each gadget lie in the open interval $(\frac{1}{2}, \frac{1}{2} + \frac{1}{n})$, the probabilities $p_i$ lie in $(0, 1)$, and the cost of $1^*$ is positive. Furthermore, for all $k$ such that $1 \leq k \leq n$, $f(k) = O(\poly(n))$, and the MDP has size $O(\poly(n))$.
\end{theorem}

\begin{proof} 
We start by showing that for all $j < k$, $|\diff'(k)| > |\diff'(j)|$. It suffices to show for all $k$ that $|\diff'(k+1)| > |\diff'(k)|$.
Recall that in the graph without the gadgets, $\diff(k+1) = (p_{k+1}-S_k)\diff(k)$. 
Thus in this new graph, $\diff'(k+1) = \left(\prod_{i=1}^{f(k)} q_i\right)(p_{k+1}-S_k)\diff(k)$ if $S_k = 1$, and $\diff'(k+1) = \left(\prod_{i=1}^{f(k)} r_i\right)(p_{k+1}-S_k)\diff(k)$ if $S_k = 0$.
$\diff'(k+1) \geq (\frac{1}{2})^{f(k+1)} \diff(k+1)$, and $\diff'(k) \leq (\frac{1}{2} + \frac{1}{n})^{f(k)}\diff(k)$. Therefore, 

$$\left|\frac{\diff'(k+1)}{\diff'(k)}\right| \geq \left| \frac{(\frac{1}{2})^{f(k+1)} \diff(k+1)}{(\frac{1}{2} + \frac{1}{n})^{f(k)} \diff(k)} \right| = \frac{(\frac{1}{2})^{f(k+1)}|p_{k+1} - S_k|}{(\frac{1}{2} + \frac{1}{n})^{f(k)}}$$
We can now plug in our formula for $f(\cdot)$ and show that this ratio is at least 1. 
\begin{align*}
    \frac{(\frac{1}{2})^{f(k+1)}}{(\frac{1}{2} + \frac{1}{n})^{f(k)}} &\geq \frac{(\frac{1}{2})^{f(k+1)}}{(\frac{1}{2} + \frac{1}{n})^{\log_{(\frac{1}{2} + \frac{1}{n})} ((\frac{1}{2})^{f(k+1)} \cdot (\frac{1}{3}))}}\\
    &= \frac{(\frac{1}{2})^{f(k+1)}}{(\frac{1}{2})^{f(k+1)} \cdot (\frac{1}{3})}
\end{align*}
Thus, since $\frac{1}{3} < |p_{k+1} - S_k|$, we have
\begin{align*}
    \left|\frac{\diff'(k+1)}{\diff'(k)}\right| &\geq \frac{(\frac{1}{2})^{f(k+1)}|p_{k+1}-S_k|}{(\frac{1}{2})^{f(k+1)} \cdot (\frac{1}{3})} \\
    &> \frac{(\frac{1}{2})^{f(k+1)}}{(\frac{1}{2})^{f(k+1)}}\\
    &= 1
\end{align*}
We have thus shown that $|\diff'(k)| > |\diff'(j)|$ for all $j < k$, as desired. To show that $f(1) = O(\poly(n))$, we observe that 
\begin{align*}
    f(k) &= \log_{\left(\frac{1}{2} + \frac{1}{n}\right)}{\left(\left(\frac{1}{2}\right)^{f(k+1)} \cdot \frac{1}{3}\right)}\\
    &= \frac{\log((\frac{1}{2})^{f(k+1)} \cdot \frac{1}{3})}{\log(\frac{1}{2} + \frac{1}{n})}\\
    &= \frac{f(k+1) \log(\frac{1}{2}) + \log(\frac{1}{3})}{\log(\frac{1}{2} + \frac{1}{n})}\\
    &= f(k+1) \frac{\log \frac{1}{2}}{\log(\frac{1}{2} + \frac{1}{n})} + O(1)
\end{align*}
We can analyze the coefficient $\frac{\log \frac{1}{2}}{\log(\frac{1}{2} + \frac{1}{n})}$.
$$\frac{\log \frac{1}{2}}{\log(\frac{1}{2} + \frac{1}{n})} = \frac{\log \frac{1}{2}}{\log \frac{1}{2} + \log(1 + \frac{2}{n})} = \left(1 - \log\left(1 + \frac{2}{n}\right) \right)^{-1}$$
Since $\frac{x}{1+x} \leq \log(1 + x) \leq x$ for $x > -1$, we have that $\frac{1}{n} \leq \log(1 + \frac{2}{n}) \leq \frac{2}{n}$ for $n > 1$. Plugging this in, we have
$f(k) = \frac{f(k+1)}{1 - \Theta(\frac{1}{n})} + O(1)$.
Thus, $f(1) = O(\poly(n))$.
Since $f(k)$ is monotonically non-increasing, $f(k) = O(\poly(n))$ for all $k$ such that $1 \leq k \leq n$.
\end{proof}

\section{Connections to the Simplex Algorithm} 
Finding an optimal policy in an MDP can be formulated as a linear program (LP). First, we present a matrix encoding of the MDP given by Hansen in \cite[Definition 2.1.2]{hansen2012worst}. Let the MDP have $N$ states and $m$ actions. Let $\bf{e} \in \mathbb{R}^N$ be an all-one vector. Let $\bf{J}, \bf{P} \in \mathbb{R}^{m \times N}$. $\bf{J}$ represents the adjacency matrix, where for each action $a$, $\bf{J}_{a,i} = 1$ if $a \in A_i$, and $\bf{J}_{a,i} = 0$ otherwise. $\bf{P}$ represents the probabilities associated with the various actions. $\bf{P}_{a,i}$ is the probability of ending up in state $i$ from action $a$. Let $\bf{c} \in \mathbb{R}^m$ represent the rewards of the actions. That is, $\bf{c}_a$ is the reward of taking action $a$.

We can solve the following linear program to obtain the optimal value $y_i$ of each state $i$:
\begin{align*}
    &\text{minimize } \bf{e}^T \bf{y}\\
    &\text{subject to } (\bf{J} - \bf{P})\bf{y} \geq \bf{c}
\end{align*}
This is equivalent to the dual LP in \cite{friedmann2011subexponential}.

Several switching rules for policy iteration have been shown to be equivalent to pivot rules for the simplex algorithm. We show that two of our lower bounds, for Simple Policy and Improvement and Difference Policy Improvement, imply equivalent smoothed lower bounds for the simplex algorithm using Bland's and Dantzig's pivot rules respectively.

We note that the perturbations in our MDP setting translate to non-standard perturbations in the simplex setting. In the celebrated result by Spielman and Teng that the simplex method has polynomial smoothed complexity \cite{spielman2004smoothed}, all entries in $\bf{A}$ and $\bf{y}$ are perturbed in a linear program formulated as follows:
\begin{align*}
    &\text{maximize } \bf{z}^T \bf{x}\\
    &\text{subject to } \bf{A}\bf{x} \leq \bf{y}
\end{align*}
In our MDP formulation, this would mean perturbing $(\bf{J} - \bf{P})$ and $\bf{c}$.
With such perturbations, deterministic actions could become probabilistic, and perturbations of zero entries in $\bf{J}$ could create new edges between states. 

Our smoothed MDP lower bounds translate to semi-smoothed simplex lower bounds, where weights can be perturbed but the general structure must be preserved (e.g., zero entries stay zero). 
More precisely, we do not perturb the adjacency matrix $\bf{J}$ at all.
We perturb only the nonzero and non-one entries of $\bf{P}$. For each row $a$ representing a random action, let $i$ be the first state with nonzero $\bf{P}_{a,i}$. We define $\bf{P}_{a,i} = 1 - \sum_{\substack{j \in A_i \\ j \neq i}} \bf{P}_{a,j}$. We perturb all $\bf{P}_{a,j}$ for $j \in A_i$ and $j \neq i$. This ensures that the probabilities associated with each random edge sum to 1.
We perturb all nonzero entries of $\bf{c}$.
We say that an LP parameterized by $\bf{e}, \bf{J}, \bf{P}, \bf{c}$ that is perturbed in this way is \textit{MDP-smoothed}.

Hansen shows that Bland's pivoting rule is equivalent to making the first improving switch according to some fixed permutation of the edges \cite[Section 5.8]{hansen2012worst}. This is exactly the simple policy improvement algorithm, where the edges are ordered according to the vertices' numbers.
Thus, our result in \autoref{subsec:simple} implies a $2^n$ lower bound on the number of iterations for the simplex algorithm using Bland's rule, even when the probabilities and rewards can be perturbed.

\begin{theorem}
The worst-case MDP-smoothed complexity of the simplex algorithm with Bland's pivoting rule for LPs with dimension $N$, number of constraints $O(N)$, and allowed perturbations of up to $\frac{1}{2}$, is $2^{\Omega(N)}$. 
\end{theorem}

\begin{proof}
We consider the LP formulation of the basic graph with $n$ states from \autoref{fig:basicgraph} with a cost of 1 incurred upon reaching the sink $1^*$. We use its equivalent maximization MDP as formulated in this section, translating the cost $1$ of the sink $1^*$ to a negative reward $-1$. We fix $p_k = \frac{1}{2}$ for every probability. 
Structurally-preserving perturbations of up to $\frac{1}{2}$ yield probabilities in the open interval $(0,1)$ and a positive cost (or negative reward) of the sink $1^*$, while preserving all other aspects of the MDP. 

Thus, by \autoref{thm:simple}, simple policy iteration requires $2^n$ iterations in the worst case.
While the basic graph has $n$ states, it has $2n + 3$ vertices including the sinks and $0'$. Each state has at most 2 actions. Thus, letting $N = 2n + 3$, the equivalent LP has dimension $N$ and $O(N)$ constraints, and the running time is $2^{\Omega(N)}$.
\end{proof}

There is in fact an exponential smoothed lower bound for Bland's pivoting rule under all zero-preserving perturbations of $(\bf{J} - \bf{P})$ and $\bf{c}$, rather than our more structured MDP smoothing. 
Spielman notes in a lecture\footnote{Lecture notes: \url{http://www.cs.yale.edu/homes/spielman/BAP/lect14.pdf}} that the Klee-Minty cube is robust under zero-preserving perturbations, yielding an exponential lower bound.

Fearnley and Savani \cite{fearnley2015complexity} show that Dantzig's pivot rule corresponds to policy iteration where the action with the greatest appeal is switched. Their definition of appeal is exactly our $\diff$. Thus, the simplex algorithm with Dantzig's pivot rule is equivalent to the difference policy improvement algorithm.
Our result in \autoref{subsec:difference} implies an equivalent semi-smoothed result for the simplex algorithm with Dantzig's pivot rule.

\begin{theorem}
The worst case MDP-smoothed complexity of the simplex algorithm with Dantzig's pivot rule for LPs with dimension $N$, number of constraints $O(N)$, and allowed perturbations of up to $\frac{1}{\Omega(N)}$, is $2^{\Omega(\sqrt{N})}$.
\end{theorem}

\begin{proof}
Consider the LP equivalent to the basic graph with $n$ states and gadgets $g(f(v))$ added between each min-vertex $v$ and its two children.
For each gadget, we set the probabilities $q_i$ and $r_i$ to $\frac{1}{2} + \frac{1}{2n}$. For each random vertex $k$, we set the probability $p_k$ to $\frac{1}{2} + \frac{1}{2n}$. We set the cost of $1^*$ to 1.
Perturbations of up to $\frac{1}{2n}$ ensure that the perturbed $q_i$, $r_i$, $p_k$, and cost of $1^*$ satisfy the conditions in \autoref{thm:difference}. 
Thus, difference policy iteration takes $2^n$ iterations.

The number of vertices (and thus dimension of the LP) is $N \leq 2f(1)n + n + 3$, since each min-vertex has two gadgets with at most $f(1)$ vertices, there are $n +1$ random vertices, and there are two sinks.
Since $f(1) = O(n^2)$, $N = O(n^2)$.
The number of actions is at most $2N$, since the construction is a 2-action MDP. 
Thus, the number of constraints is $O(N)$, and for any perturbations up to $\frac{1}{n} \geq \frac{1}{\Omega(N)}$, the number of iterations is $2^{\Omega(\sqrt{n})}$.
\end{proof}

\section{Conclusion}
Many greedy algorithms have fast runtime in practice, yet exponential lower bounds in the worst case.
In recent years, smoothed analysis has emerged as a popular way to reconcile this gap between practical efficiency and theoretical complexity -- a smoothed upper bound shows that the hard worst-case instances are sparse enough that small random perturbations yield polynomial runtime in expectation, suggesting that these hard instances rarely appear in practice.
Policy iteration, as one such algorithm with this gap between theory and practice, may seem like a fitting candidate for a smoothed upper bound. 
However, our results show the contrary: under a natural smoothed model, several common variants of policy iteration have subexponential or exponential lower bounds. 

Our main and most involved result is that Howard's PI (Greedy PI) requires at least subexponentially many iterations even in the smoothed model, even when the perturbations are chosen arbitrarily (rather than randomly) within a certain inverse polynomial range.
As a corollary, we obtain an exponential lower bound on the number of iterations required by Howard's PI under the reachability criterion in the worst case, without perturbations.
We also extend results from \cite{melekopoglou1994complexity} to show that Simple PI and Topological PI take at least exponential time even under very large perturbations and even for reachability MDPs ; we also show that Difference PI takes at least subexponential time under inverse polynomial perturbations.

Several interesting open questions are raised by these results.
One natural direction for future work is to investigate where such lower bounds are not possible. Which perturbations yield polynomial expected runtime -- in our model, do constant perturbations suffice? While we focused on the total reward, average reward, and reachability criteria, the discounted reward criterion is also popular. We suspect that similar results hold for discount rates $\gamma$ that are exponentially close to 1, since the behavior of such such discounted MDPs is similar to the total reward.
We have not examined if our robust construction
applies in this case. On the other hand, if $1-\gamma$ is at least inverse polynomial then we know that Greedy PI converges in polynomial time by the results of \cite{Ye,HansenMZ}.

For the reachability criterion, we showed a lower bound for Howard's PI only in the worst case. Can our result for the reachability criterion be extended to the smoothed and/or robust model? 
In the case of several PI variants that switch a single state in each iteration, Simple PI, Topological PI, and Difference PI, the bounds hold for reachability MDPs in the robust (and smoothed) model.  Are there similar smoothed/robust lower bounds for other single-switch policy iteration variants, such as the Random-Facet and Random-Edge switching rules?

\section{Acknowledgments}
This research was supported in part by NSF Grants CCF-2107187, CCF-1763970, and CCF-2212233, by JPMorgan Chase \& Co, by LexisNexis Risk Solutions, and by the Algorand Centres of Excellence programme managed by Algorand Foundation. Any opinions, findings, and conclusions or recommendations expressed in this material are solely those of the authors.

\newpage
\bibliographystyle{alpha}
\bibliography{sources}

\end{document}